\documentclass[11pt]{article}
\usepackage{jheppub}

\usepackage[bb=stix, cal=boondoxupr]{mathalpha}
\usepackage{booktabs, tabularx, longtable}
\usepackage{mathtools}

\usepackage{tikz}
\usetikzlibrary{arrows.meta, positioning, shadows, calc}

\usepackage{amsthm}
\newtheorem{thm}{Theorem}
\newtheorem{lemma}{Lemma}

\definecolor{myGreen}{HTML}{B6CE91}
\definecolor{myBlue}{HTML}{91B6CE}
\definecolor{myPurple}{HTML}{CE91B6}

\newenvironment{eqnalign}{\equation\aligned}{\endaligned\endequation}

\renewcommand{\d}[2][]{\mathrm{d}^{#1}{#2}}

\DeclareMathOperator{\U}{U}
\DeclareMathOperator{\SU}{SU}
\DeclareMathOperator{\SO}{SO}
\DeclareMathOperator{\Sp}{Sp}
\newcommand{\adj}{\mathrm{adj}}
\newcommand{\rep}[2][]{\mathbf{\underline{#2}^{#1}}}
\newcommand{\repbar}[2][]{\mathbf{\underline{\overline{#2}}^{#1}}}
\newcommand{\repss}[2]{\mathbf{\underline{#1}_{#2}}}    

\newcommand{\I}{\mathrm{I}}         
\newcommand{\II}{\mathrm{I\!I}}     
\newcommand{\Lat}{\mathfrak{L}}
\newcommand{\lat}{\mathfrak{l}}
\newcommand{\Gram}{\mathcal{G}}
\newcommand{\gram}{\mathcal{g}}
\renewcommand{\u}{\mathcal{u}}

\DeclareMathOperator{\Aut}{Aut}
\DeclareMathOperator{\tr}{tr}
\DeclareMathOperator{\diag}{diag}
\DeclareMathOperator{\rank}{rank}
\DeclareMathOperator{\nullity}{nullity}

\DeclareMathOperator{\sign}{sign}   

\newcommand{\Z}{\mathbb{Z}}
\newcommand{\Q}{\mathbb{Q}}
\newcommand{\R}{\mathbb{R}}

\newcommand{\calA}{\mathcal{A}}
\newcommand{\calD}{\mathcal{D}}
\newcommand{\calF}{\mathcal{F}}
\newcommand{\calH}{\mathcal{H}}
\newcommand{\calJ}{\mathcal{J}}
\newcommand{\calL}{\mathcal{L}}
\newcommand{\calM}{\mathcal{M}}
\newcommand{\calN}{\mathcal{N}}
\newcommand{\calQ}{\mathcal{Q}}
\newcommand{\calR}{\mathcal{R}}
\newcommand{\calZ}{\mathcal{Z}}

\newcommand{\ch}{{\mathrm{ch}}}   
\newcommand{\irr}{{\mathrm{irr}}} 
\newcommand{\BPS}{{\mathrm{BPS}}}

\title{A finite 6d supergravity landscape from anomalies}

\author[\clubsuit\spadesuit]{Yuta Hamada}
\author[\clubsuit]{and Gregory J.\ Loges}
\affiliation[\clubsuit]{
    Theory Center, IPNS, High Energy Accelerator Research Organization (KEK), \\
    1-1 Oho, Tsukuba, Ibaraki 305-0801, Japan
}
\affiliation[\spadesuit]{
    Graduate Institute for Advanced Studies, SOKENDAI, 1-1 Oho, Tsukuba, Ibaraki 305-0801, Japan
}

\emailAdd{yhamada@post.kek.jp}
\emailAdd{gloges@post.kek.jp}

\preprint{KEK-TH-2741}

\abstract{
    6d supergravities with non-abelian gauge group are subject to many consistency conditions.
    While the absence of local gauge and gravitational anomalies allows for infinitely many models, we show that those conditions stemming from the absence of both local and global anomalies together are strong enough to leave only finitely many consistent models.
    To do this we distill the consequences of anomaly cancellation into a high-dimensional linear program whose dual can be efficiently studied using standard techniques.
    We obtain a universal bound on the number of tensor multiplets $T\leq 11\cdot 273 = 3003$ and show that this leads to a finite landscape of consistent non-abelian models.
    Interestingly, the model which saturates this bound has gauge group $[E_8\times F_4\times(G_2\times\SU(2))^2]^{273}$, which bears a striking resemblance to the model which saturates the bound $T\leq 193$ for F-theory constructions.
}

\begin{document}

\maketitle
\flushbottom

\section{Introduction}
\label{sec:introduction}

It has long been appreciated that amongst all low-energy effective theories only a very small fraction are compatible with the principles of quantum gravity.
Several general arguments have been made that the quantum gravity landscape is in fact finite (e.g.\ see~\cite{Vafa:2005ui,Acharya:2006zw,Grimm:2021vpn,Hamada:2021yxy}), and this has been borne out very concretely in high dimensions in the presence of supersymmetry.
For example, in ten dimensions minimal supersymmetry and absence of local gauge and gravitational anomalies is enough to essentially fix the gauge group~\cite{Green:1984sg} and the short list of possibilities exactly matches those found in string theory, a phenomenon which has come to be known as \emph{string universality} or the \emph{string lamppost principle}.
Ambitiously, one may hope to show that this continues all the way to four dimensions and in the absence of low-energy supersymmetry.

Understanding string universality has since progressed, most successfully for $d>6$ and with supersymmetry~\cite{Kim:2019ths,Montero:2020icj,Cvetic:2020kuw,Hamada:2021bbz,Bedroya:2021fbu}.
Continuing down to $d=6$ presents a new challenge since it is at this point that the minimal number of supercharges drops to eight and matter supermultiplets in arbitrary representations first appear.
Much work has gone into understanding the subtleties of the consistency conditions for these models~\cite{Kumar:2010ru,Seiberg:2011dr,Monnier:2018nfs,Lee:2019skh,Tarazi:2021duw,Dierigl:2022zll,Lee:2022swr,Kim:2024tdh,Kim:2024hxe,Lockhart:2025lea}.

Recently in~\cite{Kim:2024hxe} it was argued by Kim, Vafa and Xu that the landscape of consistent 6d, $\calN=(1,0)$ supergravities with non-abelian gauge group is in fact finite.
This is thanks to a sharp upper bound on the number of massless modes derived using the presence of a so-called \emph{H-string} in the BPS spectrum.
In particular, they find the bounds $T\leq 193$ and $\rank G\leq 480$, both of which are saturated by known F-theory constructions; this provides additional evidence for string universality.
A similar bound of $T\leq 567$ was also recently achieved by Birkar and Lee~\cite{Birkar:2025rcg} as a consequence of techniques developed to bound Hodge numbers of elliptic Calabi-Yau 3-folds.
In the present work we give an argument for the finiteness of this 6d landscape which is complementary to both of these.
Our upper bound on $T$ will come from analyzing the structure of anomaly lattices subject to constraints coming from the absence of Dai-Freed anomalies; consequently our working assumptions are somewhat different than those of~\cite{Kim:2024hxe} and~\cite{Birkar:2025rcg}.

Dai-Freed~\cite{Dai:1994kq} global anomalies arise when properly accounting for the fact that differential forms cannot fully capture all of the features of an (anti-)self-dual $2$-form gauge field.
These anomalies are detected by placing the theory on a non-trivial curved background, such as the lens spaces that we consider in this work.
In the context of 6d supergravity, Dai-Freed anomalies have recently been studied in~\cite{Basile:2023zng,Dierigl:2025rfn} and shown to be quite discerning.

\medskip

The logical organization of our argument is summarized in figure~\ref{fig:flowchart}.
After reviewing the basic structure of 6d supergravities and their local anomalies in section~\ref{sec:preliminaries}, section~\ref{sec:lattice-lower-bound} is dedicated to distilling the consequences of absence of Dai-Freed anomalies into a linear program which provides a strong lower bound on $T$ in terms of other data.
In section~\ref{sec:finiteness} the various bounds are brought together to prove finiteness and obtain upper bounds on $T$ and the gauge group in progressively more restrictive settings.
Finally, we conclude with a discussion in section~\ref{sec:discussion}.
Appendix~\ref{app:orbits} contains a proof of the uniqueness of $(3;1^T)$ for the gravitational anomaly vector and appendices~\ref{app:eta-invariants} and~\ref{app:column-groups} contain data supporting the calculations of the main text.

\begin{figure}
    \centering
    \begin{tikzpicture}[
            arr/.style = {very thick, -Triangle},
            arr-join/.style = {very thick, rounded corners},
            box/.style = {rectangle, draw, semithick, minimum height=9mm, minimum width=17mm, drop shadow, align=center},
        ]

        \node [box, fill=myPurple] (finite) at (0,0) {Finite landscape};

        \node [box, fill=myGreen, above = 2.5 of finite] (T-bound) {$T\leq 3003$};
        \draw [arr] (T-bound) -- (finite);

        \node [box, fill=myGreen, left = 1 of T-bound] (lat-bound) {$T\geq f^\ast(N)$};
        \draw [arr] (lat-bound) -- (T-bound) node (arr1) [midway, below] {\S\ref{sec:upper-bounds}};
        
        \node [box, fill=myGreen, above = 1 of arr1.north] (grav-bound) {$\Delta + 29T \leq 273$};
        \draw [arr-join] (grav-bound) |- (T-bound);
        
        \node [box, fill=myGreen, left = 1 of lat-bound] (bi-fixed) {if $b_i^2\leq -1$, $b_i^0=0$};
        \draw [arr] (bi-fixed) -- (lat-bound) node [midway, below] {\S\ref{sec:linear-programming}};
        
        \node [box, fill=myGreen, left = 1 of bi-fixed] (anom-global) {Dai-Freed anomalies};
        \draw [arr] (anom-global) -- (bi-fixed) node (arr2) [midway, below] {\S\ref{sec:Dai-Freed-anomalies}};

        \node [box, fill=myGreen, above = 1 of arr2.north] (b0-fixed) {$b_0=(3;1^T)$};
        \draw [arr-join] (b0-fixed) |- (bi-fixed);

        \node [box, fill=myBlue, above = 1 of b0-fixed] (asm-b0) {$b_0\in\I_{1,T}$ or $b_0/2\in\II_{1,T}$ is primitive};
        \draw [arr] (b0-fixed) -- (asm-b0);
        \draw [arr] (asm-b0) -- (b0-fixed)  node [midway, right] {\S\ref{app:orbits}};
        
        \node [box, fill=myGreen, below left = 1.5 and 0.3 of arr1.north] (fixed-T) {Finite number of models for fixed $T$};
        \draw [arr-join] (fixed-T) -| (finite);
        
        \node [box, fill=myGreen, below = 1.5 of fixed-T] (anom-probe) {String probes \& anomaly inflow};
        \draw [arr] (anom-probe) -- (fixed-T) node (arr3) [midway, right] {\S\ref{sec:string-probes}};
        
        \node [box, fill=myBlue, above left = 0.75 and 4 of arr3.west] (asm-index) {$[\Gamma:\Gamma^\BPS]\leq C<\infty$};
        \node [box, fill=myGreen, left = 4 of arr3.west] (J-exists) {$\exists\calJ$ with $\calJ^2>0$ and $\calJ\cdot b_i>0$};
        \node [box, fill=myGreen, below left = 0.75 and 4 of arr3.west] (bsqr-bound) {$b_i^2 \leq h(\Delta_i)$};
        \draw [arr-join] (J-exists) -| (fixed-T);
        \draw [arr-join] (asm-index) -| ($(J-exists.east) + (0.5,0.5)$) |- ($(J-exists.east) + (1,0)$);
        \draw [arr-join] (bsqr-bound) -| ($(J-exists.east) + (0.5,-0.5)$) |- ($(J-exists.east) + (1,0)$);

    \end{tikzpicture}
    \caption{Logical organization of the finiteness argument. Assumptions are in blue and the absence of local gauge anomalies is used implicitly throughout.}
    \label{fig:flowchart}
\end{figure}
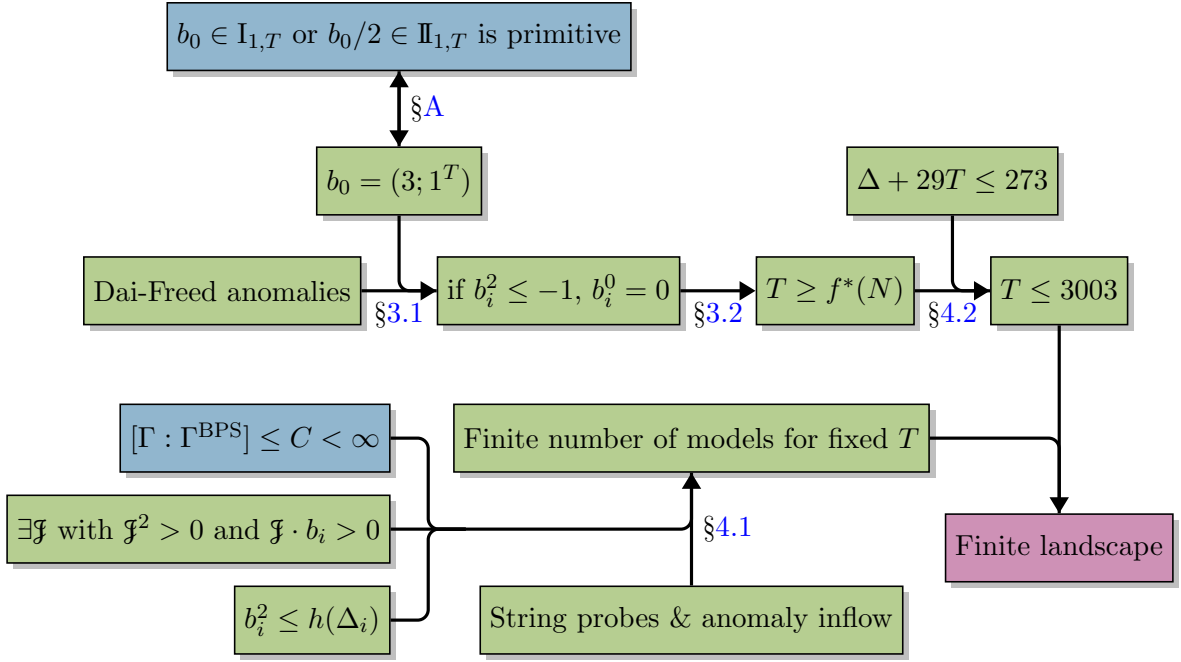

\section{Preliminaries}
\label{sec:preliminaries}

In this section we review the structure of 6d, $\calN=(1,0)$ supergravities and set our notation.
First of all, with minimal supersymmetry there are four types of supermultiplets -- gravity, tensor, vector and hyper -- and a model $\calM$ amounts to a choice for the number of tensor multiplets $T$, gauge group $G$ and charged (i.e.\ non-trivial) hypermultiplet representations $\calH_\ch$.
As we will review below, the number of hypermultiplets transforming in the trivial representation is determined by anomaly cancellation and therefore implicitly fixed by these data.
The gauge group, $G=\prod_i G_i$, we take to be non-abelian and reserve $i,j$ to index data corresponding to the simple factors.
There are $T+1$ tensor fields, one in the gravity multiplet which is self-dual and one in each tensor multiplet which is anti-self-dual.
We reserve $\alpha,\beta\in\{0,1,\ldots,T\}$ to index these, with ``$0$'' always referring to the self-dual field.
To each gauge factor is associated a vector $b_i$ and there is a corresponding vector\footnote{The notation $a \coloneqq -b_0$ is common.} $b_0$ associated with the gravity multiplet; we refer to these collectively using the indices $I,J$.
These vectors live in a space of signature $(1,T)$ with inner product $\Omega$ which we abbreviate using $u\cdot v \coloneqq u^\alpha\Omega_{\alpha\beta}v^\beta$ and $u^2\coloneqq u\cdot u$.
Finally, the scalars in the tensor multiplets are written as $\calJ \in \SO(1,T)/\SO(T)$; these parametrize the gauge kinetic terms $\propto \calJ\cdot b_i$, which should be non-negative to avoid ghosts.

Local anomalies are captured by the $8$-form polynomial $I_8$ which receives contributions from the chiral matter in all multiplets~\cite{Alvarez-Gaume:1983ihn}.
In order to be canceled via the Green-Schwarz-West-Sagnotti (GSWS) mechanism, we must be able to write~\cite{Green:1984sg,Green:1984bx,Sagnotti:1992qw}
\begin{equation}
    I_8 + \frac{1}{2}\Omega_{\alpha\beta}X_4^\alpha\wedge X_4^\beta \stackrel{!}{=} 0 \,,
\end{equation}
where the $4$-form $X_4$ takes the following form:
\begin{equation}
    X_4^\alpha = -\frac{b_0^\alpha}{2}\tr\calR^2 + \frac{1}{2}\sum_i\frac{b_i^\alpha}{\lambda_i}\tr\calF_i^2 \,.
\end{equation}
The normalization constants $\lambda_i$ are listed in table~\ref{tab:group-constants}.
Having added the appropriate couplings between $B^\alpha$ and $X_4$ in the action, the Bianchi identities read
\begin{equation}
    \d{H^\alpha} = X_4^\alpha \,.
\end{equation}
Imposing the above cancellation requires that the irreducible terms ($\tr\calR^4$ and $\tr\calF_i^4$) in $I_8$ vanish, leading to the constraints
\begin{equation}
    H - V + 29T = 273 \,, \qquad \sum_R n_R^iB_R^i - B_\adj^i = 0 \,,
\end{equation}
while matching the reducible terms fixes all of the inner products for the $b_I$:
\begin{eqnalign}
\label{eq:inner-products}
    b_0^2 &= 9-T \,, &
    b_0\cdot b_i &= \frac{\lambda_i}{6}\Big(\sum_Rn_R^iA_R^i - A_\adj^i\Big) \,,\\
    b_i^2 &= \frac{\lambda_i^2}{3}\Big(\sum_Rn_R^iC_R^i - C_\adj^i\Big) \,, &\quad
    b_i\cdot b_j &= \lambda_i\lambda_j\sum_{R,S}n_{R,S}^{i,j}A_R^iA_S^j \quad (i\neq j) \,.
\end{eqnalign}
The non-negative coefficients $n_R^i$ ($n_{R,S}^{i,j}$) give the total number of hypermultiplets transforming in the representation $R$ ($R\otimes S$) of $G_i$ ($G_i\times G_j$), and the group-theoretic indices $A_R^i$, $B_R^i$ and $C_R^i$ are defined via
\begin{equation}
    \tr_R\calF_i^2 = A_R^i\tr\calF_i^2 \,, \qquad
    \tr_R\calF_i^4 = B_R^i\tr\calF_i^4 + C_R^i(\tr\calF_i^2)^2 \,,
\end{equation}
where ``$\tr$'' with no subscript refers to the trace in the fundamental representation.
Finally, introducing $\Delta\coloneqq H_\ch - V \leq H - V$ the absence of local gravitational anomalies amounts to the upper bound
\begin{equation}
    \label{eq:Delta-29T-upper-bound}
    \Delta + 29T \leq 273 \,.
\end{equation}
If this bound is satisfied then neutral hypermultiplets can always be added to make up the difference and cancel the gravitational anomaly exactly.

\begin{table}[t]
    \centering
    \begin{tabular}{>$c<$|*{8}{>$c<$}}
        \toprule
        G_i & \SU(N) & \SO(N) & \Sp(N) & E_6 & E_7 & E_8 & F_4 & G_2 \\ \midrule
        \lambda_i & 1 & 2 & 1 & 6 & 12 & 60 & 6 & 2\\
        h_i^\vee & N & N-2 & N+1 & 12 & 18 & 30 & 9 & 4 \\
        \bottomrule
    \end{tabular}
    \caption{Normalization constants and dual-Coxeter numbers for the classical and exceptional groups.}
    \label{tab:group-constants}
\end{table}

\subsection{Anomaly and string charge lattices}
\label{sec:lattices}

A key figure in later sections will be the anomaly lattice $\Lat$ generated by the anomaly vectors $b_I$, as well as its sublattice $\lat\subseteq\Lat$ generated by the $b_i$ only.
While not immediately obvious since $A_R,B_R,C_R\in\Q$, it can be shown that the inner products in~\eqref{eq:inner-products} are all integers~\cite{Kumar:2010ru}.
That is, $\Lat$ is an integer lattice with corresponding Gram matrix $\Gram_{IJ} \coloneqq b_I\cdot b_J$ determined only by the chiral spectrum.
We will also make use of $\gram_{ij}\coloneqq b_i\cdot b_j$, the Gram matrix for $\lat$, as well as the integer vector $\u_i\coloneqq b_0\cdot b_i$ which are related to $\Gram$ via
\begin{eqnalign}
\label{eq:gram-decomp}
    \Gram = \left(\begin{array}{c|c}
        9-T & \u^\top\\ \hline
        \u & \gram
    \end{array}\right) \,.
\end{eqnalign}
We write the signature of $\Lat$ as $(n_+^\Lat, n_0^\Lat, n_-^\Lat)$.
The non-trivial requirement for local anomalies to be canceled via the GSWS mechanism is that $n_+^\Lat\leq 1$ and $n_-^\Lat \leq T$.
While $n_\pm^\Lat$ give the number of positive and negative eigenvalues of $\Gram$, null directions are more subtle.
Since $\Lat$ lives in $\R^{1,T}$ it is clear that $n_+^\Lat + n_0^\Lat \leq 1$ and $n_0^\Lat + n_-^\Lat \leq T$.
In contrast, $n_0^\Gram \coloneqq \nullity(\Gram) \geq n_0^\Lat$ can in principle be arbitrarily large if many of the $b_I$ are linearly dependent.
The above discussion applies also to $\lat$ and $\gram$, and by the Cauchy interlacing theorem we have $n_\pm^\lat\leq n_\pm^\Lat \leq n_\pm^\lat + 1$.

In addition to the bounds on the signature of $\Lat$ coming from the GSWS mechanism, it has been shown that $\Lat$ must be a sublattice of the unimodular string charge lattice $\Gamma=\Gamma_{1,T}$ of signature $(1,T)$~\cite{Seiberg:2011dr}.
In contrast to Euclidean unimodular lattices, the classification of indefinite unimodular lattices is quite simple and the lattice $\Gamma$ is essentially uniquely determined.
For each $T\geq 0$ there is a unique odd unimodular lattice of signature $(1,T)$, while even unimodular lattices of signature $(1,T)$ exist only for $T\equiv 1\mod{8}$ but again are unique.
We denote these lattices by $\I_{1,n}$ and $\II_{1,8n+1}$, respectively.
The lattice $\I_{1,n}$ consists of vectors $\Z^{n+1}$ with quadratic form $\Omega=\diag(1,(-1)^n)$.
The lattice $\II_{1,8n+1}$ can be described as the set of vectors
\begin{eqnarray}
    \Big\{x\in\Z^{8n+2}\cup\big(\tfrac{1}{2}+\Z\big)^{8n+2}\;\Big|\; \sum\nolimits_\alpha x^\alpha \equiv 0\mod{2}\Big\} \,,
\end{eqnarray}
again with quadratic form $\Omega=\diag(1,(-1)^{8n+1})$.
Thanks to its uniqueness, it can equally well be described as $\II_{1,8n+1} = U\oplus (-E_8)^{\oplus n}$, where $U$ is the unique even unimodular lattice of signature $(1,1)$ consisting of vectors $\Z^2$ with inner product $\Omega=(0\; 1; 1\; 0)$ and $(-E_8)$ is the $E_8$ root lattice with reversed signature.
The requirement $\Lambda\hookrightarrow\Gamma$ by itself will not play a major role in later sections since $T \geq n_-^\Gram + 5$ is a sufficient condition for such an embedding to exist~\cite{Nikulin:1980}.

Finally, one can show from the expressions of equation~\eqref{eq:inner-products} that for each $b_i$ we have $b_i^2\equiv b_0\cdot b_i\mod{2}$, i.e.\ $b_0$ is a characteristic vector of $\Lat$~\cite{Monnier:2017oqd}.
In addition, a careful analysis of the Green-Schwarz terms shows that $b_0$ must also be a characteristic vector of the full string charge lattice~\cite{Monnier:2018nfs}.
We will assume that $b_0\in\I_{1,T}$ or $b_0/2\in\II_{1,T}$ is primitive;\footnote{A vector $v$ in the lattice $\Gamma$ is primitive if there is no integer $k>1$ such that $v\in k\Gamma$.} note that in the even case $b_0$ cannot be both primitive and characteristic since $\II_{1,8n+1}$ contains no such vectors.\footnote{For example, consider the six vectors in $\II_{1,1}$ of norm $8$: $(3;1),(-3;-1)\in 2\II_{1,1}$ are characteristic but not primitive and $(3;-1),(-3;1),(\frac{9}{2};\frac{7}{2}),(-\frac{9}{2};-\frac{7}{2})\in\II_{1,1}$ are primitive but not characteristic (e.g.\ since $(3;-1)\cdot(\frac{1}{2};-\frac{1}{2})=1 \not\equiv 0\mod{2}$ and $(\frac{9}{2};\frac{7}{2})\cdot(1;1)=1 \not\equiv 0\mod{2}$).}
This allows us to always fix $b_0=(3;1^T)\in\I_{1,T}$ or $b_0=(3;1^T)\in 2\II_{1,T}$ by a suitable integral change of basis, which is, of course, the form distinguished by 6d F-theory constructions where $b_0$ is identified with the canonical class of the base manifold.
We provide a proof of this in appendix~\ref{app:orbits} by showing that there is only a single orbit under $\Aut(\I_{1,n})$ or $\Aut(\II_{1,8n+1})$ consisting of vectors with the required properties.

\subsection{Model decomposition}
\label{sec:model-decomposition}

By permuting the $b_i$ the Gram matrix $\gram$ may be brought to block-diagonal form,
\begin{eqnalign}
    \gram = \bigoplus_a \gram_a \,,
\end{eqnalign}
where each of the $\gram_a$ is an irreducible matrix.
Associated to these $\gram_a$ are the vectors $\u_a$ from which $\u$ is built.
This gives us a natural decomposition of a model $\calM$ into a collection $\{\calM_a^\irr\}$ of irreducible models, indexed by $a$.
In general it is not the case that $\lat$ is a direct sum of the lattices $\lat_a$ since all null directions must be identified (recall that $n_0^\lat \leq 1$ as argued below~\eqref{eq:gram-decomp}).
In doing so all anomaly vectors with $b_i^2=0$ must lie on the future light cone so that $\calJ\cdot b_i>0$ is possible.
With this decomposition we can write
\begin{eqnalign}
    \Delta = \sum_a\Delta_a \,, \quad n_\pm^\lat = \sum_a n_\pm^{\lat_a} \,, \quad n_0^\lat \leq \sum_a n_0^{\lat_a} \,, \quad n_0^\gram = \sum_a n_0^{\gram_a} \,, 
\end{eqnalign}
where $\Delta_a$, $n_\pm^{\lat_a}$ etc.\ have the obvious meanings.
We emphasize that each $\calM_a^\irr$ need not be consistent on their own for $\calM$ to be consistent; for example, one can easily find examples that have $\Delta_a\gg 273$ accompanied by a sufficient number of compensating negative contributions to $\Delta$.

\begin{table}
    \centering
    \begin{tabular}{*{6}{>$c<$}}
        \toprule
        G_a  & \calH_{\ch,a}       & \Delta_a & \gram_a & \u_a^\top & \text{Number}\\ \midrule
        E_8  & \text{--}           & -248     & [-12]   & [-10]  & N^{(12)}   \\
        E_7  & \text{--}           & -133     & [-8]    & [-6]  & N^{(8)}      \\
        E_7'  & \frac{1}{2}\rep{56} & -105     & [-7]    & [-5]  & N^{(7)}      \\
        E_6  & \text{--}           & -78      & [-6]    & [-4]  & N^{(6)}      \\
        F_4  & \text{--}           & -52      & [-5]    & [-3]  & N^{(5)}      \\
        \SO(8) & \text{--}           & -28      & [-4]    & [-2]  & N^{(4)}      \\
        \SU(3) & \text{--}           & -8       & [-3]    & [-1]  & N^{(3)}      \\
        G_2\times\SU(2) & \frac{1}{2}(\rep{1}\oplus\rep{7},\rep{2}) & -9 & \begin{bsmallmatrix}
            -3 & 1\\ 1 & -2
        \end{bsmallmatrix} & [-1 \; 0]  & N^{(32)} \\
        \SO(7)\times\SU(2)^2 & \frac{1}{2}(\rep{8},\rep{2},\rep{1})\oplus\frac{1}{2}(\rep{8},\rep{1},\rep{2}) & -11 & \begin{bsmallmatrix}
            -3 & 1 & 1\\ 1 & -2 \\ 1 & & -2
        \end{bsmallmatrix} & [{-1} \; 0 \; 0]  & N^{(322)} \\
        \SO(8)\times\SU(2)^3 & \scalebox{0.87}{$\begin{aligned}
            &\textstyle\frac{1}{2}(\repss{8}{v},\rep{2},\rep{1},\rep{1}) \oplus \frac{1}{2}(\repss{8}{s},\rep{1},\rep{2},\rep{1})\\[-5pt]
            &\textstyle\qquad \oplus\frac{1}{2}(\repss{8}{c},\rep{1},\rep{1},\rep{2})
        \end{aligned}$} & -13 & \begin{bsmallmatrix}
            -3 & 1 & 1 & 1 \\ 1 & -2 \\ 1 & & -2 \\ 1 & & & -2
        \end{bsmallmatrix} & [-1 \; 0 \; 0 \; 0]  & N^{(3222)} \\
        \SU(8) & \rep{36} & -27 & [-1] & [-1]  & N^{(1\perp)} \\
        \bottomrule
    \end{tabular}
    \caption{Examples of irreducible models with minimal numbers of hypermultiplets. All but the last two correspond to the non-Higgsable clusters identified in~\cite{Morrison:2012np}.}
    \label{tab:irr-model-examples}
\end{table}

There are eleven note-worthy irreducible models which we list in table~\ref{tab:irr-model-examples}; with so few hypermultiplets they comprise the main obstacle to showing finiteness.
All but two correspond to the non-Higgsable clusters identified in~\cite{Morrison:2012np}, and with hindsight the remaining two are needed to obtain sufficiently strong bounds.
Extending the list to include additional refinements can in principle lead to improved bounds, but we keep only this minimal set for simplicity.
By consulting the hypermultiplet content of simple factors with $b_i^2<0$ in table~\ref{tab:simple-bi2-neg}, one can quickly check that $b_i\cdot b_j = 0$ whenever $b_i^2,b_j^2\leq -3$, and when $b_i^2=-2$, $b_i\cdot b_j$ is non-zero for at most one $b_j$ with $b_j^2=-3$.
This means that one can unambiguously count the number of irreducible submatrices of $\gram$ that match those in the first ten rows of table~\ref{tab:irr-model-examples}.
Let us denote their number by
\begin{eqnalign}
    N \coloneqq \begin{bmatrix}
        N^{(12)} & N^{(8)} & N^{(7)} & N^{(6)} & N^{(5)} & N^{(4)} & N^{(3)} & N^{(32)} & N^{(322)} & N^{(3222)}
    \end{bmatrix}^\top \in \Z^{10} \,,
    \label{eq:N}
\end{eqnalign}
where $N^{(s)}$ counts the number of irreducible submatrices of $\gram$ with diagonal entries $-s$.
Having done so, let $N^{(1\perp)}$ be the maximal number of mutually orthogonal vectors with $b_i^2=b_0\cdot b_i=-1$ that are also orthogonal to all of the vectors tallied by $N$.
For example,
\begin{eqnalign}
    \Gram &= \begin{bsmallmatrix}
        9-T & 9 & 3 & -1 & -1 & -2\\
        9 & 7 & 6 & 2 & 2 & 2\\
        3 & 6 & 1\\
        -1 & 2 & & -3\\
        -1 & 2 & & & -3\\
        -2 & 2 & & & & -4
    \end{bsmallmatrix} \quad \longrightarrow \quad & N^{(4)} &= 1 \,,\;\; N^{(3)} = 2 \,,\\
    \Gram &= \begin{bsmallmatrix}
        9-T & -1 & 0 & -1 & -1\\
        -1 & -3 & 1\\
        0 & 1 & -2 & 1\\
        -1 & & 1 & -1 \\
        -1 & & & & -1
    \end{bsmallmatrix} \quad \longrightarrow \quad & N^{(32)} &= 1 \,,\;\; N^{(1\perp)} = 1 \,,\\
\end{eqnalign}
with all other $N^{(s)}$ vanishing.
Also define $N_a$ and $N_a^{(1\perp)}$ analogously for irreducible models, so $N\coloneqq \sum_a N_a$ and $N^{(1\perp)} = \sum_a N_a^{(1\perp)}$.
Also useful will be the constant vector
\begin{eqnalign}
    \underline{\Delta} \coloneqq \begin{bmatrix}
        -248 & -133 & -105 & -78 & -52 & -28 & -8 & -9 & -11 & -13
    \end{bmatrix}  \in \Z^{10} \,,
    \label{eq:Delta_underline}
\end{eqnalign}
consisting of the smallest values of $\Delta_a$ over irreducible models with Gram matrix $\gram_a$ that \emph{exactly} match the first ten rows of table~\ref{tab:irr-model-examples}.

\begin{table}[t]
    \centering
    \begin{tabular}{*{1}{>$c<$}*{1}{>$l<$}*{3}{>$c<$}}
        \toprule
        G_i    & \quad\calH_{\ch,i}                                & b_i^2 & b_0\cdot b_i & \text{Note}     \\
        \midrule
        \SU(N) &   (2N)\times\rep{N}                               & -2    &  0                             \\
        \SU(N) &  (N+8)\times\rep{N} \oplus \rep{N(N-1)/2}         & -1    &  1           & N \geq 4        \\
        \SU(N) &  (N-8)\times\rep{N} \oplus \rep{N(N+1)/2}         & -1    & -1           & N \geq 8        \\
        \SU(2) &     10\times\rep{2}                               & -1    &  1                             \\
        \SU(3) & \text{--}                                         & -3    & -1                             \\
        \SU(3) &     12\times\rep{3}                               & -1    &  1                             \\
        \SU(6) &    15\times\rep{6} \oplus \frac{1}{2}\rep{20}     & -1    &  1                             \\[7pt]
        \SO(N) &  (N-8)\times\rep{N}                               & -4    & -2           & N\geq 8         \\
        \SO(N) &  (N-7)\times\rep{N} \oplus (2^{\lfloor\frac{10-N}{2}\rfloor})\times\rep{2^{\lfloor\frac{N-1}{2}\rfloor}}
                                                                        & -3    & -1           & 7\leq N \leq 12 \\
        \SO(N) &  (N-6)\times\rep{N} \oplus (2\cdot 2^{\lfloor\frac{10-N}{2}\rfloor})\times\rep{2^{\lfloor\frac{N-1}{2}\rfloor}}
                                                                        & -2    &  0           & 7\leq N \leq 13 \\
        \SO(N) &  (N-5)\times\rep{N} \oplus (3\cdot 2^{\lfloor\frac{10-N}{2}\rfloor})\times\rep{2^{\lfloor\frac{N-1}{2}\rfloor}}
                                                                        & -1    &  1           & 7\leq N \leq 12 \\[7pt]
        \Sp(N) & (2N+8)\times\rep{2N}                              & -1    &  1                             \\[7pt]
        E_6    &      k\times\rep{27}                              & k-7   & k-5          & k\leq 6         \\
        E_7    & \frac{k}{2}\times\rep{56}                         & k-8   & k-6          & k\leq 7         \\
        E_8    & \text{--}                                         & -12   & -10                            \\
        F_4    &      k\times\rep{28}                              & k-5   & k-3          & k\leq 4         \\
        G_2    & (3k+1)\times\rep{7}                               & k-3   & k-1          & k\leq 2         \\
        \bottomrule
    \end{tabular}
    \caption{All simple groups and choices for charged hypermultiplets which give $b_i^2<0$.}
    \label{tab:simple-bi2-neg}
\end{table}

\section{A lattice lower bound}
\label{sec:lattice-lower-bound}

The requirement that an embedding $\Lat\hookrightarrow\Gamma$ exists does not by itself constrain $T$ very much.
A sufficient condition for such an embedding to exist (for $\Gamma$ either even or odd) is simply that $\Lat$ be integral with $n_+^\Lat \leq 1$ and $n_-^\Lat \leq T - 5$~\cite{Nikulin:1980}.
Therefore without some extra information about the anomaly lattice, we cannot hope to improve substantially upon the trivial lower bound $T\geq n_-^\Lat$.
In section~\ref{sec:Dai-Freed-anomalies} we review the role of quadratic refinements and Dai-Freed anomalies and extract an additional property that the embedding $\Lat\hookrightarrow\Gamma$ must have to give a consistent model.
This extra information is then used in section~\ref{sec:linear-programming} to construct a linear program whose optimal solution gives a bound of the form $T\geq f^\ast(N)$, where typically $f^\ast(N)\gg n_-^\Lat$.

\subsection{Quadratic refinements and Dai-Freed anomalies}
\label{sec:Dai-Freed-anomalies}

The purpose of this section is to use the absence of Dai-Freed global anomalies to argue that $b_i\in\lat$ with $m_i\coloneqq -b_i^2 \geq 1$ must have $b_i^0 = 0$.
This then implies that such vectors must take the form\footnote{To see this subtract $\sum_{\alpha=1}^T b_i^\alpha = m-2$ from $\sum_{\alpha=1}^T (b_i^\alpha)^2 = m$.}
\begin{eqnalign}
\label{eq:bi-form}
    b_i &= (0;-1,1^{m_i-1},0^{T-m_i}) \quad\text{or}\quad (0;2,1^{m_i-4},0^{T-m_i+3}) \,,
\end{eqnalign}
(the latter only for $m_i\geq 4$) up to a permutation of the spatial components.

\medskip

The modern treatment of chiral $p$-form fields on a space $X$ realizes the dynamical fields as boundary modes of a bulk theory defined on a space $Y$ of one higher dimension with $\partial Y=X$.
In addition, the correct language to use is that of differential cohomology which, as we will review briefly below, necessitates the use of quadratic refinements in order to make sense of the bulk theory on $Y$.
Dai-Freed anomalies capture to what extent physical quantities depend on the choice of $Y$ and should vanish if the theory is to be sound.
We adopt the notation of~\cite{Hsieh:2020jpj}.

The differential cohomology is constructed in the following way.
First of all, let $C_p(X)$ denote the space of $p$-chains ($p$-dimensional subspaces) of $X$ and $Z_p(X)\subset C_p(X)$ the space of closed subspaces.
Next, let $C^p(X,\R)$ be the space of cochains, i.e.\ linear functions $C_p(X)\to\R$, and similarly for $C^p(X,\Z)$, $Z^p(X,\R)$ and $Z^p(X,\Z)$.
Defining $\delta:C^p(X,\R)\to C^{p+1}(X,\R)$ via
\begin{equation}
    \int_N\delta\textsf{x} = \int_{\partial N}\textsf{x} \,,
\end{equation}
clearly $\delta^2=0$ and we can define the cohomology groups $H^p(X,\R) = Z^p(X,\R)/\delta C^{p-1}(X,\R)$ in the usual way, and similarly for $H^p(X,\Z)$.
The data relevant to describing a $p$-form gauge field is the $(p+1)$-form field strength $\textsf{F}\in Z^p(X,\R)$ and holonomies $\chi(M)$ for all $M\in Z_p(X)$ which must be related through
\begin{equation}
    \chi(\partial N) = \exp\!\left(2\pi i\int_N\textsf{F}\right) \,.
\end{equation}
An equivalent way to package these data is through a triple $\check{A} = (\textsf{N},\textsf{A},\textsf{F})$ with $\textsf{N}\in Z^{p+1}(X,\Z)$ and $\textsf{A}\in C^p(X,\R)$, subject to the constraint
\begin{equation}
    \textsf{N} = \textsf{F} - \delta\textsf{A} \,.
\end{equation}
There remains an ambiguity
\begin{equation}
    (\textsf{N},\textsf{A},\textsf{F}) \;\;\to\;\; (\textsf{N} - \delta\textsf{n}, \textsf{A} + \delta\textsf{a} + \textsf{n},\textsf{F}) \,,
\end{equation}
with $\textsf{n}\in C^p(X,\Z)$ and $\textsf{a}\in C^{p-1}(X,\R)$, under which the holonomies
\begin{equation}
    \chi(M) = \exp\!\Big(2\pi i\int_M\mathsf{A}\Big) \,, \qquad M\in Z_p(X)
\end{equation}
are unchanged.
The differential cohomology $\check{H}^{p+1}(X)$ is the set of such triples modulo identification under this gauge transformation.
The cohomology element $[\textsf{N}]_\Z \in H^{p+1}(X,\Z)$ is invariant under the above gauge transformation and gives the quantized fluxes for the gauge field.
When $[\textsf{N}]_\Z = 0$ there is a gauge transformation which allows us to set $\textsf{N}=0$ in which case $\textsf{A}$ is actually a $p$-form and $\textsf{F} = \d{\textsf{A}}$.

One can define a product between $\check{A}\in \check{H}^{p_A+1}(X)$ and $\check{B}\in \check{H}^{p_B+1}(X)$ with
\begin{eqnalign}
    \check{A}\star\check{B} = (\textsf{N}_{\check{A}\star\check{B}},\textsf{A}_{\check{A}\star\check{B}},\textsf{F}_{\check{A}\star\check{B}}) \coloneqq (\textsf{N}_A\cup\textsf{N}_B, \textsf{A}_A\cup\textsf{N}_B + \ldots, \mathsf{F}_A\wedge\mathsf{F}_B) \in \check{H}^{p_A+p_B+2}(X) \,,
\end{eqnalign}
where $\cup$ is the usual cup product on cochains.
With this, we can define a pairing
\begin{equation}
    (\check{A}_1,\check{A}_2) \coloneqq \int_Y \textsf{A}_{\check{A}_1\star\check{A}_2} \quad \in \R/\Z \,.
\end{equation}
Here $\check{A}_i\in\check{H}^{p+2}(Y)$ and $Y$ is $(2p+3)$-dimensional.
A quadratic refinement of this pairing is a function $\calQ:\check{H}^{p+1}(Y) \to \R/\Z$ satisfying
\begin{equation}
    \calQ(\check{A}_1 + \check{A}_2) - \calQ(\check{A}_1) - \calQ(\check{A}_2) + \calQ(0) = (\check{A}_1,\check{A}_2) \,.
\end{equation}
The shift $\tilde{\calQ}(\check{A}) \coloneqq \calQ(\check{A}) - \calQ(0)$ is also a quadratic refinement with $\tilde{\calQ}(0) = 0$ by definition.

In~\cite{Hsieh:2020jpj} it was shown that the following function is a quadratic refinement for the present case of $\check{A}\in\check{H}^4(Y)$ and $Y$ a seven-dimensional spin manifold:
\begin{eqnalign}
    \calQ(\check{A}) &\coloneqq \int_Z \left(\frac{1}{2}\textsf{F}_A\wedge\textsf{F}_A - \frac{1}{4}p_1(\calR)\wedge\textsf{F}_a + 29\hat{A}_2(\calR)\right)\\
    &= -\frac{1}{2}\eta(\calD_Y) + \int_Y\left(\frac{1}{2}\textsf{A}_a\wedge\textsf{A}_a + \textsf{A}_a\wedge\textsf{F}_S - \frac{1}{4}p_1(\calR)\wedge\textsf{A}_a\right) \,.
\end{eqnalign}
In the first line above, $Z$ is any eight-dimensional spin manifold with $\partial Z=Y$.
The second line follows from the Atiyah-Patodi-Singer index theorem after splitting $\check{A}=\check{S}+\check{a}$ with $\textsf{N}_a = 0$ so that $\textsf{A}_a$ is a differential form and $\textsf{F}_a = \d{\textsf{A}_a}$.
This shows that $\calQ$ is independent of the choice of $Z$.

Using this quadratic refinement we are then able to write down the action for a bulk field on $Y$:
\begin{eqnalign}
\label{eq:bulk-action}
    -S = -2\pi\int_Y \frac{1}{2g^2}\mathsf{F}_A\wedge{\star\mathsf{F}_A} +  2\pi i\kappa\tilde{\calQ}(\check{A}) + 2\pi i\kappa(\check{A},\check{C})_Y \,.
\end{eqnalign}
$\check{C}\in\check{H}^4(Y)$ is a background field for the electric $2$-form symmetry.
For $Y$ with $\partial Y=X$ and suitable boundary conditions, the above defines a theory with self-dual ($\kappa=+1$) or anti-self-dual ($\kappa=-1$) tensor field on the six-dimensional space $X$.

Aggregating all of the contributions from chiral fields, the partition function reads
\begin{eqnalign}
    \calZ(Y) &= \exp[2\pi i\calA(Y)] \,, \qquad \calA(Y)\in\R/\Z \,,
\end{eqnalign}
where the anomaly $\calA(Y)$ can be split as
\begin{eqnalign}
    \calA(Y) &\coloneqq \calA_\text{2-forms}(Y) + \calA_\text{fermions}(Y) \,.
\end{eqnalign}
We consider the case where $Y$ is a lens spaces $L_p^7=S^7/\Z_p$ for some $p\geq 2$.
One could take more general backgrounds, but as we will see having anomalies cancel on lens spaces is already strong enough for our purposes.
As shown in~\cite{Hsieh:2020jpj}, the contribution to the anomaly from $B_\pm^\alpha$ is
\begin{equation}
    \calA_\text{2-forms}(L_p^7) = {-\tilde{\calQ}^0(\check{C})} + \sum_{\alpha=1}^T\tilde{\calQ}^\alpha(\check{C}) + \calA_\text{grav}(L_p^7)
\end{equation}
with
\begin{equation}
    \calA_\text{grav}(L_p^7) = (273-28T)\eta_0^\text{Dirac}(L_p^7) - \eta^\text{gravitino}(L_p^7) \,.
\end{equation}
Let $\check{C}_1$ correspond to a generator of $H^4(L_p^7,\Z) \cong \Z_p$ and turn on a background for the simple gauge factor $G_i$.
Then through the GSWS couplings, $\check{C}^\alpha = kb_i^\alpha\check{C}_1$ and we find
\begin{eqnalign}
    \calA_\text{2-forms}(L_p^7) &= -\tilde{\calQ}(kb_i^0\check{C}_1) + \sum_{\alpha=1}^T\tilde{\calQ}(kb_i^\alpha\check{C}_1) + \calA_\text{grav}(L_p^7) \\
    &= -\frac{kb_i^0(p^2-2) + k^2(b_i^0)^2}{2p} + \sum_{\alpha=1}^T\frac{kb_i^\alpha(p^2-2) + k^2(b_i^\alpha)^2}{2p} + \calA_\text{grav}(L_p^7) \\
    &= -\frac{k^2b_i^2 + k(p^2-2)(b_0\cdot b_i - 2b_i^0)}{2p} + \calA_\text{grav}(L_p^7) \,.
\end{eqnalign}
We have used both $b_0=(3;1^T)$ and
\begin{equation}
    \tilde{\calQ}(n\check{C}_1) = n\tilde{\calQ}(\check{C}_1) + \binom{n}{2}(\check{C}_1,\check{C}_1)
\end{equation}
for $n\in\Z$, which follows from the defining property of quadratic refinements and induction on $n$.
The $\eta$-invariants on lens spaces are known: see~\cite{APS:1975}.
The invariants for a Dirac fermion of charge $q$ and gravitino are
\begin{eqnalign}
\label{eq:lens-space-eta}
    \eta_q^\text{Dirac}(L_p^7) &= -\frac{(p^2 - 1)(p^2 + 11) + 30q(p - q)(q^2 - pq - 2)}{720p} + \Z \,,\\
    \eta^\text{gravitino}(L_p^7) &= 4\big[\eta_1^\text{Dirac}(L_p^7) + \eta_{-1}^\text{Dirac}(L_p^7)\big] - 3 \eta_0^\text{Dirac}(L_p^7) = -\frac{(p^2-1)(p^2 - 37)}{144p} + \Z \,.
\end{eqnalign}
Notice that for the degenerate case $L_1^7 = S^7$ both $\eta^\text{gravitino}(S^7)=\Z$ and
\begin{equation}
    \eta_q^\text{Dirac}(S^7) = \binom{q+1}{4} + \Z = \Z
\end{equation}
are trivial.
The fermion contributions are
\begin{eqnalign}
    \calA_\text{fermions}(L_p^7) &= \eta^\text{gravitino}(L_p^7) - T\eta_0^\text{Dirac}(L_p^7) + \eta^\text{adj}(L_p^7) - \sum_R \eta^R(L_p^7) \,,
\end{eqnalign}
where $\eta^R(L_p^7)$ for the vectors and hypermultiplets can be written in terms of $\eta_q^\text{Dirac}(L_p^7)$ as described in appendix~\ref{app:eta-invariants}.
As a check, if the background is turned off ($k=0$) then we have
\begin{equation}
    \calA(L_p^7) = \eta^\text{gravitino}(L_p^7) - (H-V+T)\eta_0^\text{Dirac}(L_p^7) + \calA_\text{grav}(L_p^7)
\end{equation}
which vanishes exactly when $H-V+29T=273$.

As argued in~\cite{Basile:2023zng}, one should restrict attention to lens spaces $L_p^7$ for which $[\d{H^\alpha}]\in \Z_p$ is trivial for at least one $\alpha$.
For example, consider the $E_8$ non-Higgsable cluster for which we compute
\begin{eqnalign}
    \calA_\text{2-forms}(L_p^7) &= \frac{12k^2 + k(p^2-2)(10+2b_i^0)}{2p} + \calA_\text{grav}(L_p^7) \,,\\
    \calA_\text{fermions}(L_p^7) &= -\frac{(p^2-1)(p^2-37)}{144p} + \eta^{\rep{248}}(L_p^7) - (521-28T)\eta_0^\text{Dirac}(L_p^7) \,,\\
    \calA(L_p^7) &= \frac{12k^2+k(p^2-2)(10+2b_i^0)}{2p} + \Z = \frac{k(6k-5-b_i^0)}{p} + \Z \,.
\end{eqnalign}
Notice that the $T$-dependence has dropped out; this is a general feature.
Suppose that $b_i^0=0$, so either $b_i=(0;2,1^8,0^{T-9})$ or $b_i=(0;-1,1^{11},0^{T-12})$.
In the former case, the relevant parameters are
\begin{equation}
    (k,p) = (1,3) \,,\; (2,3) \,,\; (3,2) \,,\; (3,3) \,,\; (3,5) \,,\; (4,3) \,,\; (4,7) \,,
\end{equation}
while for the latter they are
\begin{equation}
    (k,p) = (1,2) \,,\; (1,3) \,,\; (2,3) \,,\; (3,2) \,,\; (3,3) \,,\; (3,4) \,,\; (4,3) \,,\; (4,5) \,,
\end{equation}
and it is quickly checked that the anomaly is trivial for all $15$ of these pairs.
In contrast, if one picks $b_i = (-1;-1^3,1^{10},0^{T-13})$ then for each $k$ there is at least one allowed value of $p$ for which there is an anomaly:
\begin{eqnalign}
    k &= 1: \quad & \calA(L_4^7) &= \tfrac{1}{2} + \Z \,,\\
    k &= 2: \quad & \calA(L_3^7) &= \tfrac{1}{3} + \Z \,,\\
    k &= 3: \quad & \calA(L_4^7) &= \tfrac{1}{2} + \Z \,,\\
    k &= 4: \quad & \calA(L_3^7) &= \tfrac{2}{3} + \Z \,.
\end{eqnalign}
More generally, one may pick $b_i^0$ and then find all candidate vectors $b_i$ using backtracking on the spatial components.
We have checked out to $|b_i^0|\leq 20$ and only for $b_i^0=0$ does the anomaly vanish for all $p$.
We take this as strong evidence that $b_i^0=0$ is required since as the components of $b_i$ grow there are ever more lens spaces that are allowed backgrounds and $\calA(L_p^7)\in\frac{1}{p}\Z$ is less likely to vanish when $p$ is large.

A similar analysis holds for each row of table~\ref{tab:simple-bi2-neg}.
Computing $\calA(L_p^7)$ using the values of $b_i^2$ and $b_0\cdot b_i$ and using appendix~\ref{app:eta-invariants} to determine the contributions of the vector multiplets and hypermultiplets, one finds that only for $b_i^0 = 0$ is the model free of anomalies.
In these calculations we assume the minimal choice for global structure of the gauge group.

\subsection{Linear programming}
\label{sec:linear-programming}

Consider a set of $k$ vectors $b_i\in\Lat$ with $m_i\coloneqq -b_i^2 \geq 2$ such that for each $m_i=2$ there is exactly one $m_j=3$ for which $b_i\cdot b_j$ is nonzero and all other pairs of vectors are orthogonal.
The number of vectors of each number is captured by $N\in\Z^{10}$ (see~\eqref{eq:N} and table~\ref{tab:irr-model-examples}).
As we have seen, all of these vectors must each take the form of~\eqref{eq:bi-form} to avoid there being global anomalies.
Because at most one of their components are negative, the ways that two vectors can have non-zero components in common is quite limited.
For example, the only ways to make $b_1=(0;2,1^{m_1-4},0,\ldots)$ and $b_2=(0;-1,1^{m_2-1},0,\ldots)$ orthogonal by permuting their components are
\begin{eqnalign}
\label{eq:orthog-example}
    &\left\{\begin{array}{cccccc}
        b_1 = (0; &  2 & 1^2 & 1^{m_1-6} & 0^{m_2-3} & 0^{T-m_1-m_2+6}) \,, \\
        b_2 = (0; & -1 & 1^2 & 0^{m_1-6} & 1^{m_2-3} & 0^{T-m_1-m_2+6}) \,,
    \end{array}\right. \\[7pt]
    &\left\{\begin{array}{ccccccc}
        b_1 = (0; & 2 & 1^{m_1-6} &  1 & 1 & 0^{m_2-2} & 0^{T-m_1-m_2+5}) \,, \\
        b_2 = (0; & 0 & 0^{m_1-6} & -1 & 1 & 1^{m_2-2} & 0^{T-m_1-m_2+5}) \,,
    \end{array}\right. \\[7pt]
    &\left\{\begin{array}{cccccc}
        b_1 = (0; & 2 & 1^{m_1-4} &  0 & 0^{m_2-1} & 0^{T-m_1-m_2+3}) \,, \\
        b_2 = (0; & 0 & 0^{m_1-4} & -1 & 1^{m_2-1} & 0^{T-m_1-m_2+3}) \,,
    \end{array}\right.
\end{eqnalign}
which require $T\geq m_1 + m_2 - 6$, $T\geq m_1 + m_2 - 5$ and $T\geq m_1 + m_2 - 3$, respectively.
It is not hard to convince oneself that when all $m_i\geq 4$ the lowest that $T$ can be is $\sum_i(m_i-3)$, which can be achieved by taking combinations of non-overlapping $(2,1^{m-4})$s and blocks of the form
\begin{equation}
    \begin{array}{rrrrrrrrr}
        -1 & 1 & 1 & 1 &\;\; 1^{m_1-4} \\
        1 & -1 & 1 & 1 & & 1^{m_2-4} \\
        1 & 1 & -1 & 1 & & & 1^{m_3-4} \\
        1 & 1 & 1 & -1 & & & & 1^{m_4-4}
    \end{array}
\end{equation}
for example.
Clearly such constructions cannot work if any of $N^{(3)},\ldots,N^{(3222)}$ are non-zero, so we would like to be more systematic.

\medskip

Inspired by~\eqref{eq:orthog-example}, the idea will be to take the $k\times T$ matrix $b_i^\alpha$ ($i\in\{1,2,\ldots,k\}$ and $\alpha\in\{1,2,\ldots,T\}$) and shift our attention from the rows to the columns, which as we have just seen can only contain a very limited distribution of non-zero values.
It is clear that for each $\alpha\geq 1$, the multiset of non-zero values $\{b_i^\alpha\}_i$ can contain at most one $2$, at most one $-1$ and never a $2$ and $1$ together.
Less clear is that the number of $1$s cannot be more than three, but this can be shown by starting with a column containing four ones and trying to extend the vectors to achieve pair-wise orthogonality.
Therefore if a column contains non-zero entries they must be one of
\begin{eqnalign}
\label{eq:single-columns}
    &\begin{bmatrix}
        1
    \end{bmatrix}
    \begin{array}{l}
        m_1
    \end{array} \,, \quad 
    \begin{bmatrix}
        1 \\ 1
    \end{bmatrix}
    \begin{array}{l}
        m_1 \\ m_2
    \end{array} \,, \quad 
    \begin{bmatrix}
        1 \\ 1 \\ 1
    \end{bmatrix}
    \begin{array}{l}
        m_1 \\ m_2 \\ m_3
    \end{array} \,, \quad 
    \begin{bmatrix}
        2
    \end{bmatrix}
    \begin{array}{l}
        m_1
    \end{array} \,, \quad 
    \begin{bmatrix}
        -1 \\ 2
    \end{bmatrix}
    \begin{array}{l}
        m_1 \\ m_2
    \end{array} \,,\\
    &\begin{bmatrix}
        -1
    \end{bmatrix}
    \begin{array}{l}
        m_1
    \end{array} \,, \quad
    \begin{bmatrix}
        -1 \\ 1
    \end{bmatrix}
    \begin{array}{l}
        m_1 \\ m_2
    \end{array} \,, \quad
    \begin{bmatrix}
        -1 \\ 1 \\ 1
    \end{bmatrix}
    \begin{array}{l}
        m_1 \\ m_2 \\ m_3
    \end{array} \,, \quad 
    \begin{bmatrix}
        -1 \\ 1 \\ 1 \\ 1
    \end{bmatrix}
    \begin{array}{l}
        m_1 \\ m_2 \\ m_3 \\ m_4
    \end{array} \,,
\end{eqnalign}
where $m_i$ label the norms of the vectors to which the components in that row correspond.
With hindsight, working just with individual columns is insufficient to get a good bound when $N^{(3)},\ldots,N^{(3222)}$ are large compared to $N^{(12)},\ldots,N^{(4)}$.
To address this, columns which contain components of vectors with $m_i\in\{2,3\}$ will be grouped together.
For example, one possible column grouping is
\begin{eqnalign}
\label{eq:col-grp-ex-1}
    \begin{bmatrix}
        -1 & 1 & 1 & \\
        & & -1 & 1 \\
        2 & 1 & 1 & 1
    \end{bmatrix}
    \begin{array}{l}
        3 \\ 2 \\ m_1\geq 7
    \end{array} \,,
\end{eqnalign}
which contributes once to the tally for $N^{(32)}$.
Another example is
\begin{eqnalign}
\label{eq:col-grp-ex-2}
    \begin{bmatrix}
        & & -1 & 1 & 1 \\
        -1 & & 1 & 1 \\
        1 & -1 \\
        1 & 1 & 1 & & 1 \\
        1 & 1 & & 1 & -1
    \end{bmatrix}
    \begin{array}{l}
        3 \\ 3' \\ 2' \\ m_1\geq 5 \\ m_2\geq 5
    \end{array} \,,
\end{eqnalign}
which contributes once to both $N^{(3)}$ and $N^{(32)}$.
Notice that we must have $m_1,m_2\geq 5$ since the corresponding rows must share another non-zero component in common elsewhere in the column decomposition in order to be orthogonal.
In total there are $2961$ distinct column groups after accounting for the freedom to reorder their rows and columns.
These fall into $74$ different families like the two highlighted above: see appendix~\ref{app:column-groups} for a complete list.

\medskip

The distribution of column groups which make up the matrix $b_i^\alpha$ are subject to $37$ constraints depending on $N$.
Four come simply from the tallies for $N^{(3)}$, $N^{(32)}$, $N^{(322)}$ and $N^{(3222)}$.
Similarly, for each $m\geq 4$ the total number of $2$s and $-1$s in rows with this value of $m$ must match $N^{(m)}$ ($6$ constraints).
The remaining $27$ constraints come from aggregating contributions to the norms,
\begin{eqnalign}
\label{eq:norm-constraint}
    mN^{(m)} &= \sum_{\alpha\geq 1}^T\sum_{\substack{i \\ m_i=m}}^k(b_i^\alpha)^2 \,, &\quad m &\geq 4
    \quad (6\text{ constraints})\,,
\end{eqnalign}
and inner products amongst vectors of given norms:
\begin{eqnalign}
\label{eq:orthog-constraint}
    0 &= \sum_{\alpha\geq 1}^T\sum_{\substack{i\neq j\\ m_i=m_1\\ m_j=m_2}}^k b_i^\alpha b_j^\alpha \,, &\quad m_1 &\geq m_2\geq 4 
    \quad (21\text{ constraints})\,.
\end{eqnalign}
For example, the column group of~\eqref{eq:col-grp-ex-2} with $m_1=6$, $m_2=5$ contributes only to two of the six norm sums of~\eqref{eq:norm-constraint} ($4$ to both $m=5$ and $m=6$) and only one of the $21$ orthogonality sums of~\eqref{eq:orthog-constraint} ($1$ for $m_1=6$, $m_2=5$).
In addition, there are $30$ linear inequalities which stem from the following observation.
Fix $m_1>m_2\geq 4$ and consider a $2\times 2$ block of all $1$s such as appears in the last two rows of~\eqref{eq:col-grp-ex-2}.
Each such block requires either
\begin{eqnalign}
    \begin{bmatrix}
        -1\\ 2
    \end{bmatrix}\begin{array}{l}
        m_1 \\ m_2
    \end{array} \qquad \text{or}\qquad
    \begin{bmatrix}
        -1\\ 2
    \end{bmatrix}\begin{array}{l}
        m_2 \\ m_1
    \end{array} \qquad\text{or}\qquad\emph{both}\;\;
    \begin{bmatrix}
        -1\\ 1
    \end{bmatrix}\begin{array}{l}
        m_1 \\ m_2
    \end{array}\;\;\text{and}\;\;
    \begin{bmatrix}
        -1\\ 1
    \end{bmatrix}\begin{array}{l}
        m_2 \\ m_1
    \end{array}
\end{eqnalign}
somewhere in the decomposition to ensure those two vectors in particular are orthogonal.
Therefore we must have
\begin{eqnalign}
    \#\!\left(\begin{bmatrix}
        1 & 1 \\ 1 & 1
    \end{bmatrix}\begin{array}{l}
        m_1 \\ m_2
    \end{array}\right) \leq \#\!\left(\begin{bmatrix}
        -1\\ 2
    \end{bmatrix}\begin{array}{l}
        m_1 \\ m_2
    \end{array}\right) + 
    \#\!\left(\begin{bmatrix}
        -1\\ 2
    \end{bmatrix}\begin{array}{l}
        m_2 \\ m_1
    \end{array}\right) + 
    \#\!\left(\begin{bmatrix}
        -1\\ 1
    \end{bmatrix}\begin{array}{l}
        m_1 \\ m_2
    \end{array}\right) \,,\\
    \#\!\left(\begin{bmatrix}
        1 & 1 \\ 1 & 1
    \end{bmatrix}\begin{array}{l}
        m_1 \\ m_2
    \end{array}\right) \leq \#\!\left(\begin{bmatrix}
        -1\\ 2
    \end{bmatrix}\begin{array}{l}
        m_1 \\ m_2
    \end{array}\right) + 
    \#\!\left(\begin{bmatrix}
        -1\\ 2
    \end{bmatrix}\begin{array}{l}
        m_2 \\ m_1
    \end{array}\right) + 
    \#\!\left(\begin{bmatrix}
        -1\\ 1
    \end{bmatrix}\begin{array}{l}
        m_2 \\ m_1
    \end{array}\right) \,.
\end{eqnalign}

All of these conditions can be collected into the following primal linear program:
\begin{eqnalign}
\label{eq:primal-lin-program}
    \min_{x\in\R^{2961}} c^\top x \qquad \text{s.t.} \quad
    \left\{
    \begin{aligned}
        x &\geq 0\\
        Kx &= LN\\
        Mx &\leq 0
    \end{aligned}
    \right.
\end{eqnalign}
where $K\in\calM_{37\times2961}(\Z)$, $L\in\calM_{37\times 10}(\Z)$, $M\in\calM_{30\times 2961}(\Z)$ and $c\in\Z^{2961}$ are explicit integer matrices containing data on the $2961$ column groups and  $N \in \Z^{10}$ as in~\eqref{eq:N}.
$x$ represents the distribution of column groups and the vector $c$ counts the number of non-zero columns; for example, the entries of $c$ corresponding to column groups matching the templates in~\eqref{eq:col-grp-ex-1} and~\eqref{eq:col-grp-ex-2} are $4$ and $5$, respectively.
Writing $x^\ast(N)$ for the optimal solution of~\eqref{eq:primal-lin-program}, we therefore have the bound
\begin{eqnalign}
\label{eq:T-lower-bound}
    T - N^{(1\perp)} \geq f^\ast(N) \coloneqq c^\top x^\ast(N) \,.
\end{eqnalign}
Notice that we have relaxed $x\in\Z^{2961}$ to $x\in\R^{2961}$ in order to avoid having to wrangle with integer linear programming.
This is not much of a concession since we are ultimately interested in bounding $T$ when the components of $N$ are large.
The main advantage is that we can write $x^\ast(\lambda N)=\lambda x^\ast(N)$ for any $\lambda>0$ and thus
\begin{eqnalign}
\label{eq:f-star-linear}
    f^\ast(N) = \|N\|_1f^\ast(\Omega_N) \,, \qquad \Omega_N \coloneqq \frac{N}{\|N\|_1} \in\sigma^9 \,,
\end{eqnalign}
where $\|N\|_1$ is the $L^1$ norm.
This allows us a much better understanding of the bound since $\Omega_N$ lies in a compact domain, the standard $9$-simplex $\sigma^9 \coloneqq \{x\in\R^{10}\;|\;x\geq 0\,,\; \|x\|_1=1\}$.

\begin{table}
    \centering
    \begin{tabular}{*{13}{>$c<$}}
        \toprule
        \multicolumn{12}{c}{$v'\in V(F')$} & \mathrm{Vol}(D(v'))/\mathrm{Vol}(\sigma^9) \\
        \midrule
        {}[\hspace{-5pt} & 9               & 5               & 4               & 3             & 2              & 1             & 0               & 0              & 0                & 0                & \hspace{-5pt}]^\top & 0.488 \\[4pt]
        {}[\hspace{-5pt} & \frac{26}{3}    & \frac{14}{3}    & \frac{11}{3}    & \frac{8}{3}   & \frac{5}{3}    & 1             & \frac{1}{3}     & \frac{1}{3}    & \frac{1}{3}      & \frac{2}{3}      & \hspace{-5pt}]^\top & 0.083 \\[4pt]
        {}[\hspace{-5pt} & \frac{1}{2}     & \frac{1}{2}     & \frac{1}{2}     & \frac{1}{2}   & \frac{1}{2}    & \frac{2}{3}   & \frac{1}{2}     & 3              & \frac{9}{2}      & 6                & \hspace{-5pt}]^\top & 0.053 \\[4pt]
        {}[\hspace{-5pt} & \frac{1}{2}     & \frac{1}{2}     & \frac{1}{2}     & \frac{1}{2}   & \frac{1}{2}    & \frac{5}{6}   & 1               & 3              & 4                & 6                & \hspace{-5pt}]^\top & 0.016 \\[4pt]
        {}[\hspace{-5pt} & \frac{17}{2}    & \frac{9}{2}     & \frac{7}{2}     & \frac{5}{2}   & \frac{3}{2}    & \frac{11}{12} & \frac{1}{3}     & \frac{1}{2}    & \frac{1}{2}      & 1                & \hspace{-5pt}]^\top & 0.016 \\[4pt]
        {}[\hspace{-5pt} & \frac{1}{2}     & \frac{1}{2}     & \frac{1}{2}     & \frac{1}{2}   & \frac{1}{2}    & \frac{5}{6}   & \frac{2}{3}     & 3              & \frac{13}{3}     & 6                & \hspace{-5pt}]^\top & 0.015 \\[4pt]
        {}[\hspace{-5pt} & \frac{38}{21}   & \frac{74}{63}   & \frac{64}{63}   & \frac{6}{7}   & \frac{44}{63}  & \frac{46}{49} & \frac{74}{63}   & \frac{164}{63} & \frac{241}{63}   & \frac{323}{63}   & \hspace{-5pt}]^\top & 0.014 \\[4pt]
        {}[\hspace{-5pt} & 8               & 4               & 3               & 2             & 1              & \frac{4}{5}   & \frac{1}{3}     & \frac{1}{2}    & 1                & 2                & \hspace{-5pt}]^\top & 0.013 \\[4pt]
        {}[\hspace{-5pt} & \frac{17}{2}    & \frac{9}{2}     & \frac{7}{2}     & \frac{8}{3}   & \frac{11}{6}   & 1             & \frac{1}{3}     & \frac{1}{3}    & \frac{1}{2}      & \frac{2}{3}      & \hspace{-5pt}]^\top & 0.012 \\[4pt]
        {}[\hspace{-5pt} & \frac{347}{210} & \frac{689}{630} & \frac{601}{630} & \frac{57}{70} & \frac{85}{126} & 1             & \frac{739}{630} & \frac{167}{63} & \frac{2411}{630} & \frac{1648}{315} & \hspace{-5pt}]^\top & 0.008 \\[4pt]
        \bottomrule
    \end{tabular}
    \caption{Vertices $v'\in V(F')$ with the largest domains $D(v')\subset\sigma^9$ as determined by a Monte-Carlo estimate with $10^6$ uniformly sampled points $\Omega_N\in\sigma^9$.}
    \label{tab:polytope-vertices}
\end{table}

The nature of $f^\ast$ is further illuminated by the dual linear program, whose optimal value matches that of~\eqref{eq:primal-lin-program} thanks to strong duality:
\begin{eqnalign}
\label{eq:dual-lin-program}
    \max_{u\in\R^{30},\, v\in\R^{37}} v^\top LN \qquad \text{s.t.} \quad
    \left\{
    \begin{aligned}
        & \;\;\; v \text{ free} \,,\; u \leq 0 \,,\\
        & K^\top v + M^\top u \leq c \,.
    \end{aligned}
    \right.
\end{eqnalign}
The feasible region $F$ of the dual linear program is $F\coloneqq\{(u,v)\in\R^{30}\times\R^{37}\;|\;u\leq0\,,K^\top v + M^\top u \leq c\}$.
Importantly, $F$ is independent of $N$ so we can simply write
\begin{eqnalign}
    f^\ast(\Omega_N) &= \max_{(u,v)\in V(F)} v^\top L \Omega_N = \max_{v'\in V(F')}v'^\top \Omega_N \,,
\end{eqnalign}
where $V(F)$ is the set of vertices of the polytope $F\subset\R^{67}$ and similarly for the polytope $F'\coloneqq L^\top \Pi_vF \subset\R^{10}$, where $\Pi_v:\R^{67}\to\R^{37}$ projects out $u$.
For each vertex $v'\in V(F')$ there is a corresponding convex domain $D(v')\subset\sigma^9$ within which $f^\ast(\Omega_N)=v'^\top \Omega_N$.
The vertices corresponding to the ten largest domains are given in table~\ref{tab:polytope-vertices}.
The largest domain by far has $f^\ast(\Omega_N) = \sum_{m\geq 4}(m-3)\Omega_N^{(m)}$, the bound noticed above.
The remainder of $\sigma^9$ is covered by much smaller domains which improve on this bound when $\Omega_N^{(3)},\ldots,\Omega_N^{(3222)}$ are relatively large.
Figure~\ref{fig:triangle-bounds} shows the bounds and domains on some two-dimensional subspaces of $\partial\sigma^9$.

\begin{figure}[p]
    \centering
    \includegraphics[width=\textwidth]{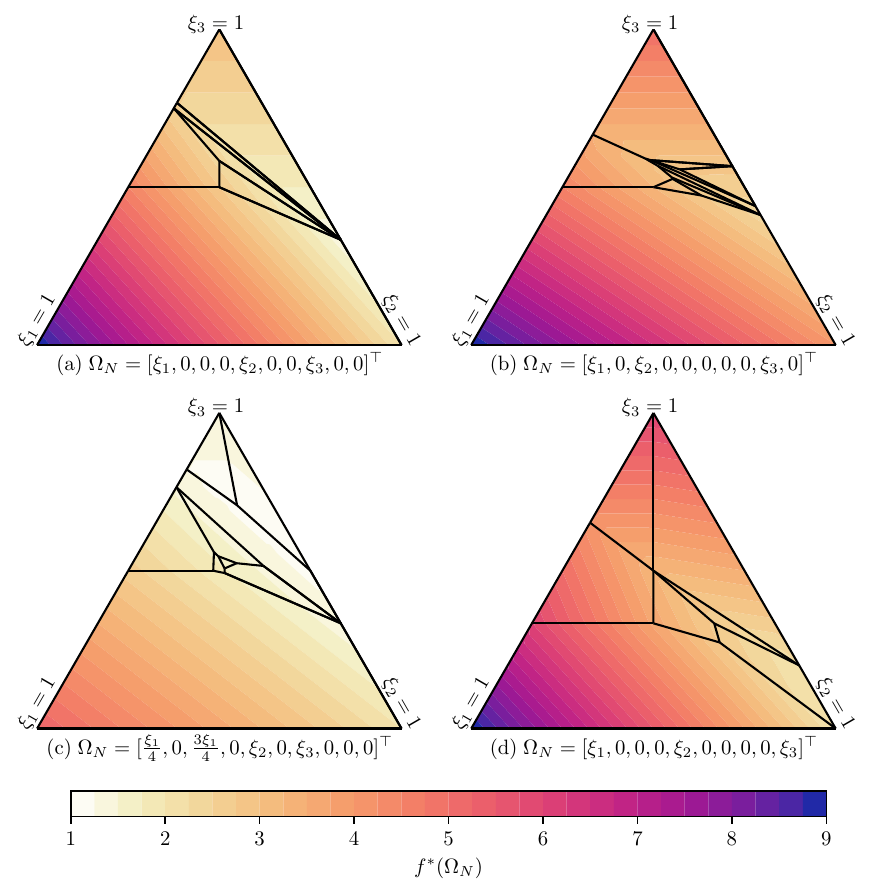}
    \caption{The function $f^\ast(\Omega_N)$ on four two-dimensional slices through the boundary of $\sigma^9$, each parametrized by $\xi_1,\xi_2,\xi_3\geq 0$ with $\xi_1+\xi_2+\xi_3 = 1$. Black lines show the partitioning into domains within which $f^\ast$ is given by different linear functions.}
    \label{fig:triangle-bounds}
\end{figure}

Given the number of inequalities defining these polytopes and the generic difficulty in enumerating the vertices of a convex polytope (e.g.\ see~\cite{Linial_1986}), it appears to be infeasible to actually obtain $V(F)$ or $V(F')$.\footnote{$F\subset\R^{67}$ is the intersection of $2961+30$ half-spaces and therefore has $\mathcal{O}(3000)$ facets (some of the defining inequalities may be redundant). Barnette's lower bound theorem~\cite{Barnette:1971,Barnette:1973} then gives $V(F)\gtrsim 3000\times 67\approx 200,\!000$, although the truth is undoubtedly much larger since $F$ is unlikely to be as highly structured as the polytopes which saturate Barnette's bound. Also, the bound extracted from McMullen's upper bound theorem~\cite{McMullen:1970} is only marginally better than the trivial bound $V(F)\leq \binom{3000}{67}\approx 10^{138}$.}
Instead, in practice we randomly sample $\Omega_N\in\sigma^9$ and solve~\eqref{eq:dual-lin-program} using the simplex method to both generate vertices and estimate the size of their domains.
Restricting to a small subset $\tilde{V}\subset V(F')$ of vertices with largest (estimated) domain sizes, we have the lower bound
\begin{equation}
    f^\ast(\Omega_N) \geq \tilde{f}^\ast(\Omega_N) \coloneqq \max_{v'\in\tilde{V}} \, v'^\top \Omega_N \,.
\end{equation}
Computer experiments suggest that if $\tilde{V}$ consists of only the top $100$ ($1000$) vertices by domain size then $\tilde{f}^\ast$ matches $f^\ast$ exactly on $\approx\! 90\%$ ($\approx\! 98\%$) of $\sigma^9$ and underestimates $f^\ast$ by at most $\approx\! 5\%$ ($\approx\! 0.5\%$) on the remainder of $\sigma^9$.
In section~\ref{sec:finiteness} below we use the estimate $\tilde{f}^\ast$ with $|\tilde{V}|=1000$ in place of $f^\ast$ when numerically minimizing and, once the extremizing model has been identified, compute $f^\ast$ exactly to confirm that they are extremal.\footnote{Code allowing one to solve these linear programs and generate vertices of the feasible region polytopes has been made available at \href{https://github.com/gloges/6d-sugra-linear-program}{github.com/gloges/6d-sugra-linear-program}.}

\medskip

We close this section with an explicit example.
Consider
\begin{equation}
\label{eq:omega-N-example}
    \Omega_N = \big[1-z \;\; 0 \;\; 0 \;\; 0 \;\; 0 \;\; 0 \;\; z \;\; 0 \;\; 0 \;\; 0 \big]^\top \,,
\end{equation}
which describes, for instance, models containing the non-Higgsable sector $E_8^{N^{(12)}}\times\SU(3)^{N^{(3)}}$ with $z=N^{(3)}/(N^{(12)}+N^{(3)})$.
There are only $31$ relevant column groups that have all $m_i$ either $3$ or $12$; nine are single columns with only $m_i=12$ (see~\eqref{eq:single-columns}) and the rest include things like
\begin{equation}
    \begin{bmatrix}
        -1 & 1 & 1
    \end{bmatrix}
    \begin{array}{l}
        3
    \end{array} \,, \quad
    \begin{bmatrix}
        -1 & 1 & 1\\
        2 & 1 & 1
    \end{bmatrix}
    \begin{array}{l}
        3 \\ 12
    \end{array} \,, \quad
    \begin{bmatrix}
        -1 & 1 & 1\\
        & -1 & 1\\
        1 & & 1\\
        1 & 1
    \end{bmatrix}
    \begin{array}{l}
        3 \\ 12 \\ 12 \\ 12
    \end{array} \,, \quad
    \begin{bmatrix}
        -1 & 1 & & 1\\
        & -1 & 1 & 1\\
        1 & & -1 & 1\\
        1 & 1 & 1
    \end{bmatrix}
    \begin{array}{l}
        3 \\ 3 \\ 3 \\ 12
    \end{array}
\end{equation}
(see appendix~\ref{app:column-groups} for the rest).
An optimal column group distribution found by solving the primal linear program for $z\in[0,\frac{1}{2}]$ is
\begin{equation}
    (1-2z) \times \left( 
    \begin{bmatrix}
        2
    \end{bmatrix}
    \begin{array}{l}
        12
    \end{array} \right) + 
    z \times \left( 
    \begin{bmatrix}
        -1 & 1 & 1\\
        2 & 1 & 1
    \end{bmatrix}
    \begin{array}{l}
        3 \\ 12
    \end{array} \right) + 
    (8-10z) \times \left( 
    \begin{bmatrix}
        1
    \end{bmatrix}
    \begin{array}{l}
        12
    \end{array} \right) \,.
\end{equation}
Similarly, for $z\in[\frac{1}{2},\frac{7}{8}]$ an optimal solution is
\begin{equation}
    (1-z) \times \left( 
    \begin{bmatrix}
        -1 & 1 & 1\\
        2 & 1 & 1
    \end{bmatrix}
    \begin{array}{l}
        3 \\ 12
    \end{array} \right) + 
    \frac{2z-1}{3} \times \left( 
    \begin{bmatrix}
        -1 & 1 & & 1\\
        & -1 & 1 & 1\\
        1 & & -1 & 1\\
        1 & 1 & 1
    \end{bmatrix}
    \begin{array}{l}
        3 \\ 3 \\ 3 \\ 12
    \end{array} \right) + 
    (7-8z) \times \left( 
    \begin{bmatrix}
        1
    \end{bmatrix}
    \begin{array}{l}
        12
    \end{array} \right) \,,
\end{equation}
for $z\in[\frac{7}{8},\frac{25}{28}]$ an optimal solution is
\begin{eqnalign}
    &\frac{25-28z}{4} \times \left( 
    \begin{bmatrix}
        -1 & 1 & 1\\
        2 & 1 & 1
    \end{bmatrix}
    \begin{array}{l}
        3 \\ 12
    \end{array} \right) + 
    \frac{8z-7}{4} \times \left( 
    \begin{bmatrix}
        -1 & 1 & 1\\
        & -1 & 1\\
        & 1 & -1\\
        2 & 1 & 1
    \end{bmatrix}
    \begin{array}{l}
        3 \\ 12 \\ 12 \\ 12
    \end{array} \right) \\
    &\qquad\qquad + 
    \frac{4z-3}{2} \times \left( 
    \begin{bmatrix}
        -1 & 1 & & 1\\
        & -1 & 1 & 1\\
        1 & & -1 & 1\\
        1 & 1 & 1
    \end{bmatrix}
    \begin{array}{l}
        3 \\ 3 \\ 3 \\ 12
    \end{array} \right) + 
    \frac{8z-7}{6} \times \left( 
    \begin{bmatrix}
        1 \\ 1 \\ 1
    \end{bmatrix}
    \begin{array}{l}
        12 \\ 12 \\ 12
    \end{array} \right) \,,
\end{eqnalign}
and for $z\in[\frac{25}{28},1]$ an optimal solution is
\begin{eqnalign}
    &\frac{28z-25}{9} \times \left( 
    \begin{bmatrix}
        -1 & 1 & & 1\\
        & -1 & 1 & 1\\
        1 & & -1 & 1
    \end{bmatrix}
    \begin{array}{l}
        3 \\ 3 \\ 3
    \end{array} \right) + 
    \frac{1-z}{3} \times \left( 
    \begin{bmatrix}
        -1 & 1 & 1\\
        & -1 & 1\\
        & 1 & -1\\
        2 & 1 & 1
    \end{bmatrix}
    \begin{array}{l}
        3 \\ 12 \\ 12 \\ 12
    \end{array} \right) \\
    &\qquad\qquad + 
    \frac{8-8z}{3} \times \left( 
    \begin{bmatrix}
        -1 & 1 & & 1\\
        & -1 & 1 & 1\\
        1 & & -1 & 1\\
        1 & 1 & 1
    \end{bmatrix}
    \begin{array}{l}
        3 \\ 3 \\ 3 \\ 12
    \end{array} \right) + 
    \frac{2-2z}{9} \times \left( 
    \begin{bmatrix}
        1 \\ 1 \\ 1
    \end{bmatrix}
    \begin{array}{l}
        12 \\ 12 \\ 12
    \end{array} \right) \,.
\end{eqnalign}
Together these give
\begin{equation}
    f^\ast(\Omega_N) = \begin{cases}
        9(1-z) & z \in [0,\tfrac{1}{2}) \,, \\
        \frac{26-25z}{3} & z \in [\tfrac{1}{2},\tfrac{7}{8}) \,, \\
        \frac{19-17z}{3} & z \in [\tfrac{7}{8}, \tfrac{25}{28}) \,, \\
        \frac{7 + 5z}{9} & z \in [\tfrac{25}{28},1] \,.
    \end{cases}
\end{equation}
In particular, we see that when $N^{(3)}<N^{(12)}$ the bound $T\geq \sum_{m\geq 4}(m-3)N^{(m)} = 9N^{(12)}$ from before is best possible, but for $N^{(3)}>N^{(12)}$ the bound is strengthened.

\section{Finiteness}
\label{sec:finiteness}

Having obtained a non-trivial lower bound on $T$ in section~\ref{sec:lattice-lower-bound}, we now turn to the question of finiteness.
In section~\ref{sec:string-probes} we use string probes to show that there are only finitely many models for any fixed $T$.
Section~\ref{sec:upper-bounds} then brings the various bounds together to derive upper bounds on $T$ and identify extremal models.

\subsection{String probes, anomaly inflow and unitarity bounds}
\label{sec:string-probes}

Introducing a string probe into a background introduces potential anomalies which must be canceled by the string's world-volume degrees of freedom via the anomaly inflow mechanism.
The central charges and levels of the current algebra were computed in~\cite{Kim:2019vuc} for a BPS string of charge $Q\in\Gamma^\BPS\subseteq\Gamma$; after removing the center-of-mass contributions, they read
\begin{eqnalign}
    c_\text{L} &= 3Q^2 + 9Q\cdot b_0 + 2 \,, \quad &
    k_\ell &= \frac{1}{2}\big(Q^2 - Q\cdot b_0 + 2\big) \,,\\
    c_\text{R} &= 3Q^2 + 3Q\cdot b_0 \,, &
    k_i &= Q\cdot b_i \,.
\end{eqnalign}
When the string tension $\mu_Q \propto Q\cdot\calJ$ and all of $c_\text{L}$, $c_\text{R}$, $k_\ell$ and $k_i$ are non-negative, imposing unitarity amounts to the bound
\begin{eqnalign}
    \sum_i \frac{k_i\dim G_i}{k_i + h_i^\vee} \leq c_\text{L} \,,
\end{eqnalign}
where $h_i^\vee$ is the dual-Coxeter number of $G_i$ (see table~\ref{tab:group-constants}).

\medskip

\begin{figure}[t]
    \centering
    \includegraphics[width=\textwidth]{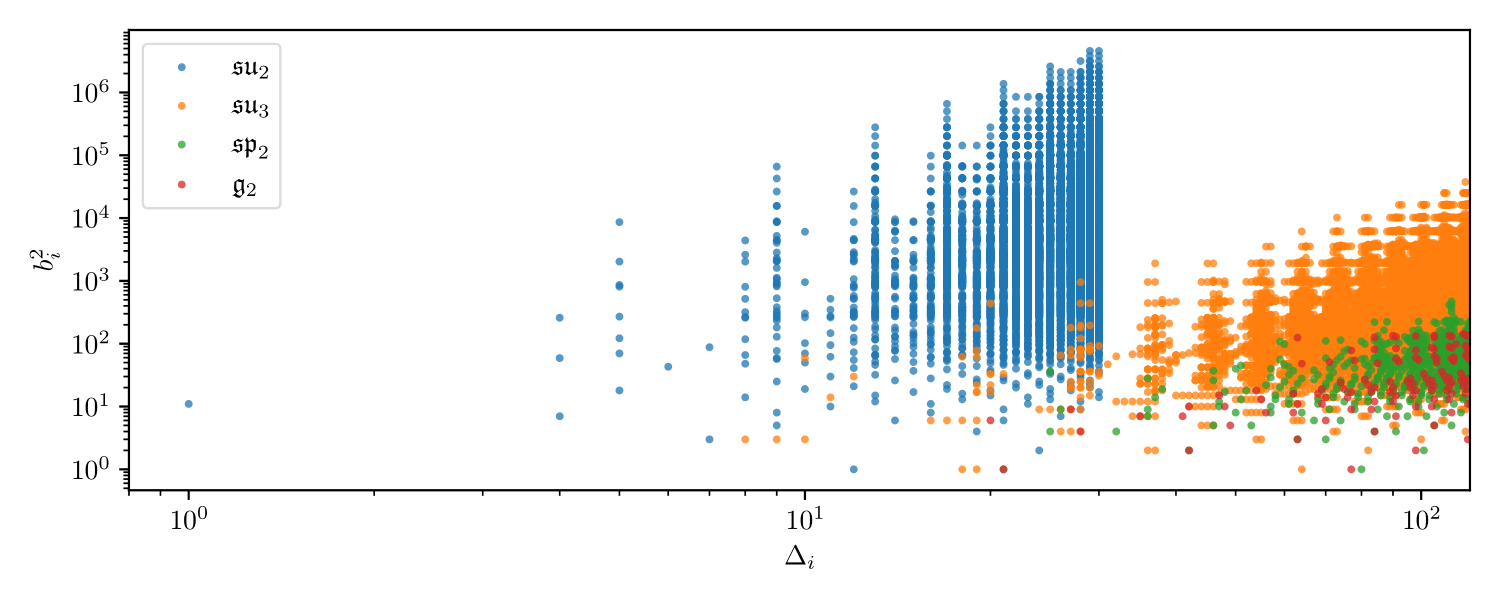}
    \caption{The growth of $b_i^2$ for low-rank non-abelian groups. For $\SU(2)$ the number of models grows very quickly and we have limited $\Delta_i\leq 30$.}
    \label{fig:bounded-bi-sqr}
\end{figure}

Along with the above consistency condition, we will need two basic facts.
The first is that there exists some constant $c$ for which $\Delta_a + cn_-^{\lat_a}$ is bounded below and all irreducible models for which it is non-positive have $n_-^{\lat_a}>0$.
Empirically, we find that $c=29$ works, but the precise value is inconsequential.
The second fact is that there exists a function $h:\Z\to\R$ such that
\begin{eqnalign}
\label{eq:bounded-diagonal}
    b_i^2 \leq h(\Delta_i)
\end{eqnalign}
for all irreducible models with simple gauge group.
This is obviously the case since achieving large $b_i^2$ requires having many hypermultiplets which then boosts the value of $\Delta_i$.
However, it turns out that $h$ must grow quickly with $\Delta_i$.
Using the numerical methods developed in~\cite{Hamada:2023zol,Hamada:2024oap}, we can generate a complete list of simple irreducible models for the lowest-rank groups with $\Delta_i\lesssim 150$ to observe the growth of $b_i^2$: see figure~\ref{fig:bounded-bi-sqr}.
The maximal values of $b_i^2$ all come from $\SU(2)$ with unusually large hypermultiplet representations, such as
\begin{eqnalign}
    G = \SU(2) \,, \qquad \calH_\ch = 2\times\rep{2} \oplus \rep{24} \,,
\end{eqnalign}
which has $b_i^2 = 131,\!941$ despite only having $\Delta_i = 25$.
Of course, with hindsight such extreme examples never appear as a component piece of a consistent model.

We now provide a short argument that for any fixed $T$ the number of consistent models is finite.
Work via contradiction; suppose that there are infinitely many distinct consistent models with the same value of $T$.
First, we claim that there are only finitely many possible anomaly lattices $\Lat$.
To see this, note that $n_-^{\lat_a}\leq T$ is bounded, so the number of irreducible components with $\Delta_a + cn_-^{\lat_a}<0$ is bounded.
This then means that $\Delta_a + cn_-^{\lat_a}$ is bounded for each irreducible model $\calM_a^\irr$ and thus thanks to~\eqref{eq:bounded-diagonal}, all of the diagonal elements of $\gram$ are bounded.
Furthermore, for fixed $b_i^2$ and $b_j^2$, the choice of $b_0\cdot b_i$ and $b_i\cdot b_j$ is finite.
Consequently, $\u$ and the off-diagonal elements of $\gram$ are both bounded by functions of $\diag\gram$. We conclude that there are only finitely many possible Gram matrices $\Gram$ and thus finitely many corresponding lattices.
Therefore we can restrict attention to an infinite subset of consistent models with $\Lat$ fixed.
Also, thanks to the assumption about the BPS charge lattice there are only finitely many possible $\Gamma^\BPS\subseteq\Gamma$, each with full rank, and $\Gamma^\BPS \cap \{x\in\R^{1,T}\;|\;x\cdot b_i>0 \,,x^2>0\}$ is non-empty because of the existence of the tensor branch scalars $\calJ$.
Fix $Q$ from this intersection, large enough so that $c_\text{L}$, $c_\text{R}$ and $k_\ell$ are all positive.
Since $c_\text{L}$ is fixed and $k_i\geq 1$ for each $i$,
\begin{eqnalign}
    \rank G \leq \sum_i\frac{\dim G_i}{1 + h_i^\vee} \leq \sum_i \frac{k_i\dim G_i}{k_i + h_i^\vee} \leq c_\text{L}
\end{eqnalign}
gives an upper bound on the size of the gauge group.
The precise bound depends implicitly on the upper bound for the index of BPS string charge lattice $\Gamma^\BPS$ in the string charge lattice $\Gamma$, $[\Gamma:\Gamma^\BPS]\coloneqq|\Gamma/\Gamma^\BPS|$, which controls the possible choices for $Q$ and thus $c_\text{L}$,\footnote{See~\cite{Montero:2022vva} for string theory examples with $[\Gamma:\Gamma^\BPS]>1$.} but regardless there are only finitely many possible gauge groups for any fixed $T$.
Further restricting to an infinite subset of models with constant gauge group, now $H_\ch\leq 273 + V - 29T = \text{const.}$ is bounded so there are only finitely many choices for charged matter and thus only finitely many distinct models, a contradiction.

\subsection{Upper bounds and extremal models}
\label{sec:upper-bounds}

In this section we use the linear programming bound to obtain an upper bound on $T$.
Together with the argument of the previous section, this proves that there are only finitely many consistent models.

First of all, every irreducible model $\calM_a^\irr$ satisfies the bound
\begin{eqnalign}
\label{eq:bounded-NHCs}
    \Delta_a \geq -27N_a^{(1\perp)} + \underline{\Delta}N_a \,.
\end{eqnalign}
For irreducible models with only a single simple factor this is true by definition of $\underline{\Delta}$~\eqref{eq:Delta_underline}.
However, the bound also holds for larger irreducible models which can have multi-charged hypermultiplets reducing the value of $\Delta_a$ compared to the sum of its constituents.
Physically this is obvious: upon Higgsing (and dropping any resulting $\U(1)$s) the value of $\Delta_a$ only decreases.
Let us consider a subset of models for which $\Omega_N\in S$ where $S$ is a subset of $\sigma^9$, $S\subseteq \sigma^9$.
Then we have
\begin{eqnalign}
\label{eq:delta-29T-lower-bound}
    \Delta + 29T &\overset{\eqref{eq:bounded-NHCs}}{\geq} -27N^{(1\perp)} + \underline{\Delta}N + 29T \,,\\
    &\overset{\eqref{eq:T-lower-bound}}{\geq} \max_{0\leq\lambda\leq 29}\big[\lambda T + (2-\lambda)N^{(1\perp)} + \underline{\Delta}N + (29-\lambda)f^\ast(N)\big] \,,\\
    &\overset{\phantom{\eqref{eq:T-lower-bound}}}{=} \max_{0\leq\lambda\leq 29}\Big\{\lambda T + (2-\lambda)N^{(1\perp)} + |\underline{\Delta}N|\left[(29-\lambda)\tfrac{f^\ast(\Omega_N)}{|\underline{\Delta}\Omega_N|} - 1 \right]\Big\} \,,\\
    &\overset{\phantom{\eqref{eq:T-lower-bound}}}{\geq} \lambda^\ast(S)\, T \,,
\end{eqnalign}
where we have used $\underline{\Delta}N\leq 0$ and
\begin{equation}
    f^*(N) = \|N\|_1 f^*(\Omega_N)
    = \frac{\underline{\Delta}N}{\underline{\Delta}\Omega_N^\ast} f^*(\Omega_N)
    = \frac{|\underline{\Delta}N|}{|\underline{\Delta}\Omega_N^\ast|} f^*(\Omega_N)\,,
\end{equation}
in the third line and introduced
\begin{equation}
\label{eq:lambda-star-defn}
    \lambda^\ast(S) \coloneqq \min\!\left\{ 2 \,,\; 29 - \frac{|\underline{\Delta}\Omega_N^\ast|}{f^\ast(\Omega_N^\ast)} \right\} \,,
\end{equation}
in the last line. Here $\Omega_N^\ast\in S$ maximizes $|\underline{\Delta}\Omega_N|/f^\ast(\Omega_N)$ over $S$.
Provided $\lambda^\ast(S)>0$, combining this with the gravitational anomaly bound then gives
\begin{equation}
    T \overset{\eqref{eq:delta-29T-lower-bound}}{\leq} \frac{\Delta + 29T}{\lambda^\ast(S)}  \overset{\eqref{eq:Delta-29T-upper-bound}}{\leq} \frac{273}{\lambda^\ast(S)}
\end{equation}
for all models with $\Omega_N\in S$.
Determining $\Omega_N^\ast(S)$ and $\lambda^\ast(S)$ is straightforward since $\underline{\Delta}\Omega_N$ is linear and $f^\ast(\Omega_N)$ is piece-wise linear.

\begin{figure}[t]
    \centering
    \includegraphics[width=\textwidth]{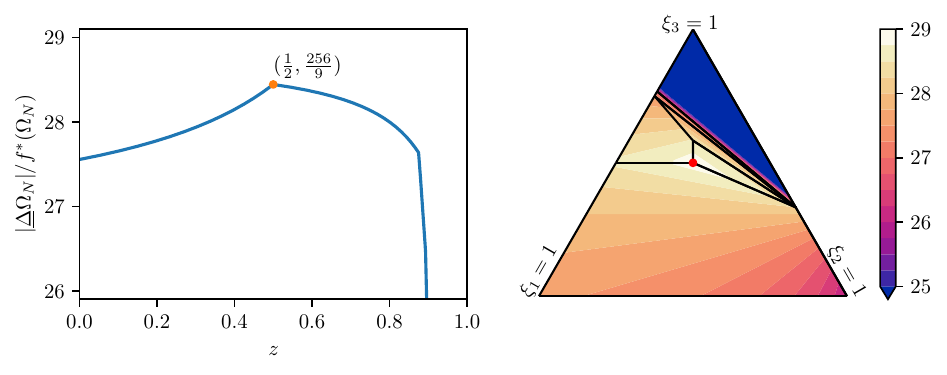}
    \caption{The ratio $|\underline{\Delta}\Omega_N|/f^\ast(\Omega_N)$ appearing in equation~\eqref{eq:lambda-star-defn} for (left) $\Omega_N$ of equation~\eqref{eq:omega-N-example} and (right) $\Omega_N=[\xi_1,0,0,0,\xi_2,0,\xi_3,0,0]^\top$ with $\xi_1+\xi_2+\xi_3 = 1$ (cf.\ figure~\ref{fig:triangle-bounds}a). The red dot indicates the point with $2\Omega_N^{(12)}=2\Omega_N^{(5)}=\Omega_N^{(32)}=\frac{1}{2}$ at which the ratio achieves its maximum value on \emph{all} of $\sigma^9$.}
    \label{fig:ratio-max}
\end{figure}

Before turning to the general case, let us return to the example where only $\Omega_N^{(12)}$ and $\Omega_N^{(3)}$ are nonzero, for which we worked out $f^\ast(\Omega_N)$ at the end of section~\ref{sec:lattice-lower-bound}.
The ratio to be maximized is shown in the left panel of figure~\ref{fig:ratio-max}; the maximum occurs at the highlighted corner where $\Omega_N^{(12)}=\Omega_N^{(3)}$.
This gives $\lambda^\ast(S) = \frac{5}{9}$ and the bound $T<492$.
We can also read off that the models which come closest to saturating this bound have equal numbers of $E_8$ and $\SU(3)$ gauge factors:
\begin{equation}
    G = E_8^\ell\times\SU(3)^\ell \,,\qquad \Delta = -256\ell \,,\qquad T\geq 9\ell \,,\qquad \Delta+29T\geq 5\ell \,.
\end{equation}
The largest value of $\ell$ compatible with the gravitational bound is $\ell=54$, for which $T=486$.

\medskip

With no restrictions on $\Omega_N$, i.e.\ $S=\sigma^9$, we find $\lambda^\ast = \frac{1}{11}$, achieved at the point with $\Omega_N^{\ast(12)}=\Omega_N^{\ast(5)}=\frac{1}{4}$ and $\Omega_N^{\ast(32)}=\frac{1}{2}$.
The right panel of figure~\ref{fig:ratio-max} shows this point within one of the low-dimensional triangular facets of $\partial\sigma^9$, showing how the three linear functions giving $f^\ast$ in the surrounding domains conspire to keep $|\underline{\Delta}\Omega_N|/f^\ast(\Omega_N)$ strictly less than $29$.
This value of $\lambda^\ast$ gives the advertised universal bound
\begin{equation}
\label{eq:T-leq-3003}
    T \leq 11\cdot 273 = 3003 \,.
\end{equation}
The simplest examples of models with the extremizing value of $\Omega_N$ have gauge groups of the form
\begin{equation}
    G = \big[E_8\times F_4\times(G_2\times\SU(3))^2\big]^\ell  \,.
\end{equation}
The maximum value of $\ell$ allowed by the gravitational anomaly is $\ell=273$, for which~\eqref{eq:T-leq-3003} is saturated.
This family bears a striking resemblance to the F-theory model with largest $T$, which has~\cite{Aspinwall:1997ye,Candelas:1997eh,Morrison:2012bb}
\begin{equation}
    G = E_8^{17}\times F_4^{16}\times (G_2\times\SU(3))^{32} \,, \qquad T = 193 \,.
\end{equation}
While our bounds are not strong enough to restrict $T$ or $\rank G$ this small, it is interesting that just with linear programming we are able to correctly identify the shape of consistent models with maximal $T$.

Some examples of the bounds and extremal models obtained when $\Omega_N$ is restricted to a subset $S\subseteq \sigma^9$ are summarized in the following table:
\begin{center}
    \begin{tabular}{*{4}{>$c<$}}
        \toprule
        S & \lambda^\ast & \lfloor 273/\lambda^\ast\rfloor & \text{Extremal model} \\
        \midrule
        \sigma^9
            & \frac{1}{11}
            & 3003
            & [E_8\times F_4\times(G_2\times\SU(3))^2]^{273}\\
        \{\Omega_N^{(5)} = 0\;\text{or}\;\Omega_N^{(32)}=0\}
            & \frac{2}{13} 
            & 1774
            & [E_8\times E_7'\times (\SO(7)\times\SU(2)^2)^2]^{136} \\
        \{\Omega_N^{(12)} = 0\}
            & \frac{1}{5}
            & 1365
            & [E_7\times(\SO(7)\times\SU(2)^2)]^{273} \\
        \{\Omega_N^{(32)} = \Omega_N^{(322)} = 0\}
            & \frac{6}{29}
            & 1319
            & [E_8\times E_7'^3\times F_4^4\times\SU(3)^8]^{35} \\
        \{\Omega_N^{(12)} = \Omega_N^{(8)} = 0\}
            & \frac{1}{3}
            & 819
            & [E_6\times\SU(3)]^{273} \\
        \{\Omega_N^{(s)}=0 ,\, s\notin\{12,3\}\}
            & \frac{5}{9}
            & 491
            & [E_8\times\SU(3)]^{54} \\
        \{\Omega_N^{(3)} = \ldots = \Omega_N^{(3222)} = 0\}
            & 1
            & 273
            & \SO(8)^{273} \\
        \bottomrule
    \end{tabular}
\end{center}
For these extremal models there is essentially zero slack in the inequalities of~\eqref{eq:delta-29T-lower-bound} and their structure is very rigid.
More general models which also include irreducible components with one or more $b_i^2>0$ will have $T$ and $N$ even more constrained.

\section{Discussion}
\label{sec:discussion}

In this work we have explored the extent to which anomaly cancellation alone constrains the space of consistent 6d, $\calN=(1,0)$ supergravities with non-abelian gauge group.
There are some features of 6d, $\calN=(1,0)$ supergravities that played an extremely minor role in our analysis.
Most notably, the $\SO(1,T)/\SO(T)$ moduli space and BPS strings only appear briefly in section~\ref{sec:string-probes} to show that $\rank G$ is bounded for fixed $T$.
The main actor in showing that the number of consistent models is finite is the high-dimensional linear program of section~\ref{sec:lattice-lower-bound} which allowed us to control the number of tensor multiplets in terms of the number of constituents with $b_i^2<0$.
In its construction the linear program takes as input that the anomaly vectors $b_I\in\Lat$ with $b_I^2<0$ must take a very particular form in the unimodular lattice $\Gamma$.
We argued for $b_i^0=0$ when $b_i^2<0$ based on the absence of Dai-Freed anomalies when placing the theory on a curved background.
This relied on some numerics to check that only for the simplest choices of the vectors $b_i$ does the anomaly vanish, and it would be worthwhile to understand if there is an underlying reason for this phenomenon.

The linear programming lower bound on the dimension $T$ can be significantly better than the simple bound coming just asking that $\Lat$ embed in $\Gamma$.
By playing this lower bound off the upper bound $\Delta+29T\leq 273$ coming from the gravitational anomaly, we obtained a universal upper bound of $T\leq 3003$.
Imposing additional constraints by hand, such as forbidding $E_8$, improves this bound significantly.
For any set of constraints we are able to identify the extremizing solution of the linear program and therefore identify the types of models which come closest to saturating our bounds.
While the resulting extreme models fall far well beyond the bounds of~\cite{Kim:2024hxe}, interestingly our methods are still able to correctly identify the ``directions'' in model space with largest $T$.

A major outstanding question is whether the landscape of consistent \emph{abelian} models is also finite.
F-theory examples with surprisingly large charges are known~\cite{Raghuram:2018hjn} and fairly simple infinite families of seemingly consistent models with three and four charges were found in~\cite{Taylor:2018khc}.
The Dai-Freed global anomalies for these families were considered in~\cite{Basile:2023zng} and it was shown that one can only rule out certain residue classes for the charges.
It remains to be seen whether there is some other consistency condition or principle which can restrict these models to a finite number as well.

\acknowledgments

Thanks to Hee-Cheol Kim and Markus Dierigl for useful conversations and thanks to Richard Borcherds for comments on appendix~\ref{app:orbits}.
This work is supported by MEXT Leading Initiative for Excellent Young Researchers Grant No.~JPMXS0320210099 [YH] and JSPS KAKENHI Grant Nos.~JP24H00976 [YH], JP24K07035 [YH], JP24KF0167 [YH and GL].

\appendix
\section{Orbits in unimodular Lorentzian lattices}
\label{app:orbits}

In this appendix, we prove that $\I_{1,n}$ contains a single orbit of primitive characteristic vectors of norm $9-n$ under $\Aut(\I_{1,n})$ and that $\II_{1,8n+1}$ contains a single orbit of primitive vectors of norm $2(1-n)$ under $\Aut(\II_{1,8n+1})$.
In the main text this is used along with the assumption that $b_0\in\I_{1,T}$ or $b_0/2\in\II_{1,T}$ is primitive to conclude that there is always an integral change of basis which brings $b_0$ to the form $(3;1^T)$.

\medskip

The low-dimensional cases are easily dealt with: one can quickly enumerate all vectors of the required norm for the lattices $\I_{1,0}$, $\I_{1,1}$ and $\II_{1,1}$, and understanding their orbits is trivial since $\Aut(\I_{1,0})\cong \Z/2\Z$ and $\Aut(\I_{1,1})\cong\Aut(\II_{1,1})\cong(\Z/2\Z)^2$ are finite groups.
For the higher-dimensional lattices we have the following theorems, generalizing somewhat to allow for vectors of different negative norms:
\begin{thm}
\label{thm:orbits-small-n}
    For every $2\leq n\leq 9$, the lattice $\I_{1,n}$ has exactly one orbit of primitive characteristic vectors of norm $9-n\geq 0$.
\end{thm}
\begin{thm}
\label{thm:orbits-n=9}
    The lattice $\II_{1,9}$ has exactly one orbit of primitive vectors of norm $0$.
\end{thm}
\begin{thm}
\label{thm:orbits-even}
    For every $n\geq 1$ and $m\geq1$, the lattice $\II_{1,8n+1}$ has exactly one orbit of primitive vectors of norm $-2m<0$.
\end{thm}
\begin{thm}
\label{thm:orbits-odd}
    For every $n\geq 3$ and $1\leq m < n$ with $m\equiv n-1\mod{8}$, the lattice $\I_{1,n}$ has exactly one orbit of primitive characteristic vectors of norm $-m<0$.
\end{thm}

\noindent
First, some comments.
The modulo-$8$ condition on $m$ in theorem~\ref{thm:orbits-odd} ensures that characteristic vectors of norm $-m$ exist in the first place.
Theorem~\ref{thm:orbits-even} is stated clearly and proved in~\cite[theorem~3.9.1]{Borcherds_1999a} (in signature $(8n+1,1)$).
Theorem~\ref{thm:orbits-odd} is known to experts but, to our knowledge, does not appear in the literature.
More focus has been placed on the counterparts of theorems~\ref{thm:orbits-even} and~\ref{thm:orbits-odd} for \emph{non-negative} norm vectors, as they display a richer structure of orbits and play a key role in classifications of classes of positive-definite lattices.
For example, orbits of primitive vectors of norm $0$, $2$ and $4$ in $\II_{1,25}$ are in one-to-one correspondence with the $24$ Niemeier lattices (including the enigmatic Leech lattice), the $121$ $25$-dimensional positive-definite even lattices of determinant $2$ and the $665$ $25$-dimensional positive-definite unimodular lattices, respectively~\cite{Borcherds_1999b}.
We give a very \emph{mechanical} proof of theorem~\ref{thm:orbits-odd} which requires the upper bound $m < n$ but the claim holds for larger $m$ as well.
However, $n\geq 3$ really is necessary; for example, there is no automorphism of $\I_{1,2}$ which relates the two primitive characteristic vectors $(1;5,5)$ and $(1;1,7)$ of norm $-49$.

\medskip

Next, we dispense with an elementary proof of theorem~\ref{thm:orbits-small-n}:

\begin{proof}[Proof of theorem~\ref{thm:orbits-small-n}]
    Let $u\in\I_{1,n}$, $2\leq n\leq 9$, be a primitive characteristic vector of norm $9-n\geq 0$ with $|u^0|>3$.
    This means that all components $u^\alpha$ are odd and $\gcd(u^\alpha)=1$.
    It suffices to show that $u$ is related through $\Aut(\I_{1,n})$ to another vector $\tilde{u}$ with $|\tilde{u}^0|<|u^0|$, since by iteratively reducing the zeroth component we must eventually arrive at a primitive characteristic vector of norm $9-n$ which has $|u^0|\leq 3$, but up to signs the only such vector is $(3;1^n)$.

    Without loss of generality we may assume that $u^0>u^1\geq \ldots\geq u^n\geq 1$ by reflecting into the positive orthant and permuting the spatial components.
    In addition, it must be that $u^1>1$ for $n<9$ and $u^1>u^9$ for $n=9$ since otherwise $u\propto(3;1^n)$, violating either $u^0>3$ or primitivity.
    Next, fix $r=(1;1^2)$ for $n=2$ and $r=(1;1^3,0^{n-3})$ otherwise;
    $r$ is a root of $\I_{1,n}$ so reflection through the plane orthogonal to $r$ is an automorphism.
    It remains only to check that
    \begin{eqnalign}
        \left|u^0 - \frac{2(u\cdot r)}{r\cdot r}\,r^0\right| < u^0 \,, \quad\text{or equivalently}\quad 0 < -u\cdot r < -(r\cdot r)u^0 \,,
    \end{eqnalign}
    having used $r^0=1$ and $r\cdot r<0$.
    The upper bound on $-u\cdot r$ follows immediately from $u^0>u^1\geq\ldots\geq u^n$.
    For $n=2$ it is easily checked that $u^1u^2\geq 9$, giving
    \begin{eqnalign}
        (u^1+u^2)^2 = (u^1)^2 + (u^2)^2 + 2u^1u^2 = (u^0)^2 - 7 + 2u^1u^2 > (u^0)^2
    \end{eqnalign}
    and therefore $-u\cdot r = -u^0 + u^1 + u^2 > 0$ as desired.
    Similarly, for $3\leq n\leq 9$ we have
    \begin{eqnalign}
        (u^1+u^2+u^3)^2 &= (u^1)^2 + (u^2)^2 + (u^3)^2 + 2(u^1u^2 + u^2u^3 + u^3u^1) \\
        &\geq (u^1)^2 + (u^2)^2 + (u^3)^2 + 6(u^3)^2 \\
        &= (9-n)(u^3)^2 + (u^1)^2 + (u^2)^2 + (u^3)^2 + (n-3)(u^3)^2 \\
        &\geq (9-n)(u^3)^2 + \sum_{i=1}^n(u^i)^2\\
        &\geq (9-n) + \sum_{i=1}^n(u^i)^2\\
        &= (u^0)^2 \,.
    \end{eqnalign}
    If $n<9$ then $u^1>1$ and either the first or last inequality is strict, and if $n=9$ then $u^1>u^9$ and one of the first two inequalities is strict.
    In either case, we again conclude that $-u\cdot r=-u^0+u^1+u^2+u^3 > 0$.
\end{proof}

Perhaps unsurprisingly, there is a qualitative change at ten dimensions where the norm of the vectors of interest changes sign.
It is also clear that following the above strategy of using reflections through roots becomes impractical, if not impossible in higher dimensions because the reflection subgroup has finite index in $\Aut(\I_{1,n})$ only for $n\leq 19$~\cite{Vinberg:1967,Vinberg:1972a,Vinberg:1972b,Vinberg:1975,Vinberg:1978}.
Therefore in order to build up to the proofs of theorems~\ref{thm:orbits-n=9}, \ref{thm:orbits-even} and~\ref{thm:orbits-odd} we will need to make use of some additional machinery.
Here we provide only a brief overview of the relevant definitions and quote some results, referring the interested reader to~\cite{Cassels_1978,Conway_1999} for more thorough treatments.

Let $L^\vee = \{x\in \operatorname{span}L\;|\; \forall y\in L\,,\; x\cdot y\in\Z\}$ denote the dual of the integer lattice $L$.
The discriminant group $G_L\coloneqq L^\vee/L$ is a finite abelian group and naturally inherits a $\Q/\Z$-valued finite bilinear form $b_L$ from $L$ and additionally a $\Q/2\Z$-valued finite quadratic form $q_L$ when $L$ is even.
For $[g],[h]\in G_L$ these finite forms are given simply by
\begin{eqnalign}
    b_L([g],[h]) &\coloneqq g\cdot h + \Z \,,\\
    q_L([g]) &\coloneqq g\cdot g + 2\Z \,,\qquad
    \text{for even $L$}\,,
\end{eqnalign}
and it is quickly checked that these are well-defined, i.e.\ independent of the particular representatives $g,h\in L^\vee$ of $[g]$ and $[h]$ chosen (critically using that $L$ is even for $q_L$).
For $L$ even, $b_L$ can be recovered from $q_L$ in the usual way:
\begin{eqnalign}
    q_L([g]+[h]) - q_L([g]) - q_L([h]) \eqqcolon 2b_L([g],[h]) \,.
\end{eqnalign}

The main move in the classification of integer lattices is to replace the ring of integers by $\R$, $\Q$, the field of $p$-adic rationals $\Q_p$ or ring of $p$-adic integers $\Z_p$.
Lattices are collected into equivalence classes called \emph{genera}, where two lattices are in the same genus if and only if they are equivalent over $\R$ and $\Z_p$ for every prime $p$.
That is, the genus of a lattice $L$ is characterized by $L\otimes\R$ and the set of $p$-adic lattices $\{L\otimes\Z_p\}$ up to $p$-adic integral equivalence.
Thanks to Sylvester's law of inertia, the signature $\sign(L)=(n_+,n_-)$ exactly characterizes when two lattices are equivalent over $\R$.
Similarly, complete sets of $p$-adic invariants allow one to systematically compute whether two $p$-adic lattices are integral equivalent.

The so-called \emph{spinor genus} is a refinement of the genus introduced by Eichler which we will not define carefully here thanks to the following well-known result~\cite{Eichler_1952}:

\begin{thm}
\label{thm:Eichler}
    For indefinite lattices of dimension at least three, a spinor genus contains exactly one integral equivalence class of lattices.
\end{thm}

\noindent
We also have the following two useful results:

\begin{thm}
\label{thm:Conway}
    If $L$ is indefinite of dimension of at least three and $G_L$ is cyclic, then the genus of $L$ contains exactly one spinor genus.
\end{thm}

\begin{thm}
\label{thm:Nikulin}
    The invariants $\sign(L)$ and $(G_L,q_L)$ determine the genus of $L$.
\end{thm}

\noindent
Theorem~\ref{thm:Eichler} follows from~\cite[Theorem 15.19]{Conway_1999}, for example, and theorem \ref{thm:Conway} falls out of the work of Nikulin on discriminant forms~\cite{Nikulin:1980} (e.g.\ see corollary~1.9.4).
For the lattices we encounter below, theorems~\ref{thm:Eichler},~\ref{thm:Conway} and~\ref{thm:Nikulin} imply that knowing the signature and discriminant form uniquely fixes the lattice.

\medskip

With these preliminaries we proceed towards proving theorems~\ref{thm:orbits-n=9}, \ref{thm:orbits-even} and~\ref{thm:orbits-odd}.
The first step is to translate the problem of understanding orbits of primitive vectors to one of understanding certain classes of lower-dimensional lattices.
The basic idea is that one can go back and forth between a primitive vector $u$ and its orthogonal complement $u^\perp$ (or $u^\perp/u$ in the case where $u$ is null).

\begin{lemma}
\label{lem:orbits-zero-norm}
    For every $n\geq 1$, there is a 1--1 correspondence between the following two sets:
    \begin{enumerate}
        \item Orbits of primitive norm-zero vectors of $\II_{1,8n+1}$ under $\Aut(\II_{1,8n+1})$.
        \item $8n$-dimensional positive-definite even unimodular lattices.
    \end{enumerate}
\end{lemma}
\begin{lemma}
\label{lem:orbits-even}
    For every $n\geq 1$ and $m\geq1$, there is a 1--1 correspondence between the following two sets:
    \begin{enumerate}
        \item Orbits of primitive vectors of norm $-2m<0$ in $\II_{1,8n+1}$ under $\Aut(\II_{1,8n+1})$.
        \item Pairs $(A,[a])$ where $A$ is an even lattice of signature $(1,8n)$ and $[a]$ is a generator of $G_A\cong\Z/2m\Z$ with $q([a]) = \frac{1}{2m} + 2\Z$.
    \end{enumerate}
\end{lemma}
\begin{lemma}
\label{lem:orbits-odd}
    For every $n\geq 1$ and $m\geq 1$ with $m\equiv n-1\mod{8}$, there is a 1--1 correspondence between the following two sets:
    \begin{enumerate}
        \item Orbits of primitive characteristic vectors of norm $-m<0$ in $\I_{1,n}$ under $\Aut(\I_{1,n})$.
        \item Pairs $(B,[b])$ where $B$ is an even lattice of signature $(1,n-1)$ and $[b]$ is a generator of $G_B\cong\Z/m\Z$ with $q([b])=\frac{m+1}{m} + 2\Z$.
    \end{enumerate}
\end{lemma}

\noindent
For a proof of lemma~\ref{lem:orbits-zero-norm} see~\cite[section~2.3]{Borcherds_1999a}; theorem~\ref{thm:orbits-n=9} follows immediately since famously $E_8$ is the unique $8$-dimensional positive-definite even unimodular lattice.
Lemmas~\ref{lem:orbits-even} and~\ref{lem:orbits-odd} are straightforward generalizations of~\cite[lemma~3.1.1]{Borcherds_1999a} (which uses signature $(8n+1,1)$) and we mimic the proof presented there, filling in some details:

\begin{proof}[Proof of lemma~\ref{lem:orbits-even}]
    To get from (1) to (2), let $u\in\II_{1,8n+1}$ be primitive with $u^2=-2m$ and take $A\coloneqq u^\perp$.
    $A$ is clearly even with the correct signature and $G_A\cong \Z/2m\Z$ since $u$ is primitive.
    Also, it is always possible to find $v\in\II_{1,8n+1}$ with $v\cdot u=1$ since $u$ is not characteristic and then define $a\coloneqq v + u/2m \in A^\vee$.
    This gives a well-defined element $[a]\in G_A$ which is independent of the choice of $v$, generates $G_A$ since $u$ is primitive. One checks that $q([a])=\frac{1}{2m}+2\Z$.
    To get from (2) to (1), form the lattice $A'\coloneqq A\oplus\langle -2m\rangle$ where $\langle -2m\rangle$ is a one-dimensional lattice generated by $u$ of norm $-2m$.
    $A'$ clearly has signature $(1,8n)$ and $G_{A'}\cong (\Z/2m\Z)^2$.
    Pick any representative $a\in A^\vee$ of $[a]\in G_A$ and define $a'\coloneqq a\oplus(u/2m)\in (A')^\vee$.
    Adding $a'$ to the lattice $A'$ results in an even integral lattice of unit determinant, hence $\II_{1,8n+1}$.
    These two maps are inverses and so we have a bijection.
\end{proof}

\noindent
The proof of lemma~\ref{lem:orbits-odd} is essentially identical after replacing $\II_{1,8n+1}$, $2m$, $A$ and $a$ by $\I_{1,n}$, $m$, $B$ and $b$ throughout, requiring only three minor changes: (i) the existence of $v\in\I_{1,n}$ with $v\cdot u=1$ is guaranteed by the primitivity of $u$ and B\'ezout's identity, (ii) the computation of $q([b]) = q([v+u/m]) = \frac{m+1}{m} + 2\Z$ uses that $u$ is characteristic, and (iii) $b'\coloneqq b\oplus(u/m)$ has odd norm so in the final step the resulting unimodular lattice is definitely $\I_{1,n}$.

\medskip

At this point we are ready to prove theorems~\ref{thm:orbits-even} and~\ref{thm:orbits-odd}.
For theorem~\ref{thm:orbits-even} we follow the proof found in~\cite{Borcherds_1999a}, adapting to signature $(1,8n+1)$:
\begin{proof}[Proof of theorem~\ref{thm:orbits-even}]
    Thanks to lemma~\ref{lem:orbits-even}, orbits correspond to pairs $(A,[a])$ with $A$ even of signature $(1,8n)$ and $[a]$ a generator of $G_A\cong\Z/2m\Z$ with $q([a])=\frac{1}{2m}+2\Z$.
    From the above discussion (theorems~\ref{thm:Eichler},~\ref{thm:Conway} and~\ref{thm:Nikulin}), the lattice $A$ is uniquely determined by these data. In fact one can check that $A\cong \langle u\rangle\oplus(-E_8)^n$ where $u^2=2m$, since this has the correct signature and discriminant form.

    Although $A$ is unique, there might be several inequivalent choices of $[a]$.
    Therefore, it remains to show that automorphisms of $A$ act transitively on the set of generators of $G_A$ which have norm $\frac{1}{2m}+2\Z$.
    We focus on $n=1$ since we only need one of the $(-E_8)$s; for larger $n$ the additional $(-E_8)$s come along for the ride.
    The map $[a]\mapsto [ka]$, $k\in\Z/2m\Z$, preserves the discriminant form exactly when $k^2\equiv 1\mod{4m}$.
    Given any $k$ satisfying this condition, we can find a vector of the form $u'=ku + 2my$, $y\in (-E_8)$, which not only has $[u'/2m] = [ku/2m] = [ka]$ but also $u'^2 = u^2 = 2m$ by choosing $y$ to have norm $y^2 = -\frac{k^2-1}{2m}\in2\Z$ (which is always possible since $E_8$ has vectors of every non-negative even norm).
    Since $u'\cdot v = 0\mod{u'^2}$ for all $v\in A$, $A\cong \langle u'\rangle\oplus u'^\perp$ and in fact $u'^\perp$ must be an even unimodular lattice of signature $(0,8)$, namely $(-E_8)$.
    Since $A\cong \langle u\rangle\oplus(-E_8)\cong\langle u'\rangle\oplus(-E_8)$, there is clearly an automorphism of $A$ which sends $u$ to $u'$ and therefore $[a]$ to $[ka]$ in $G_A$.
\end{proof}

\begin{proof}[Proof of theorem~\ref{thm:orbits-odd}]
    Thanks to lemma~\ref{lem:orbits-odd}, orbits correspond to pairs $(B,[b])$ with $B$ even of signature $(1,n-1)$ and $[b]$ a generator of $G_B\cong\Z/m\Z$ with $q([b])=\frac{m+1}{m}+2\Z$.
    From the above discussion, the lattice $B$ is uniquely determined by these data and in fact one can check that $B\cong U\oplus (-A_{m-1})\oplus (-E_8)^{r}$ ($n-m-1=8r\geq 0$) since this has the correct signature and discriminant form.
    To be explicit, $B$ has a basis $\{f_1,f_2,e_1,e_2,\ldots,e_{m-1}\}$ where $f_i$ are a basis for $U$ and the $e_i$ are a basis for $(-A_{m-1})$ with Gram matrix
    \begin{equation}
    \label{eq:An-Gram}
        \begin{bmatrix}
            -2 & 1 \\
            1 & -2 \\
            & &\ddots \\
            & & & -2 & 1\\
            & & & 1 & -2
        \end{bmatrix} \,.
    \end{equation}
    The vector $v_\ast\coloneqq e_1+2e_2+\ldots+(m-1)e_{m-1}$ has $v_\ast^2 = m(1-m)$ and $[v_\ast/m] = [b]$.

    It remains only to show that automorphisms of $B$ act transitively on the set of generators of $G_B$ which have norm $\frac{m+1}{m}+2\Z$.
    For $m=1$ this is vacuously true, so take $m\geq 2$.
    Also take $r=0$ for simplicity, larger $r$ working identically but with auxiliary $(-E_8)$s like before.
    The values of $k\in\Z/m\Z$ for which $[b]\mapsto[kb]$ preserves the discriminant form are those with $k^2\equiv 1\mod{2m}$ for $m$ even or $k^2\equiv 1\mod{m}$ for $m$ odd.\footnote{The first few non-trivial ($k\neq \pm1$) pairs are $(m,k)=(12,\pm5),\,(15,\pm4),\,(20,\pm9),\,(21,\pm8),\,(24,\pm7)$.}
    Given any $k$ satisfying the appropriate condition, define $v_\ast'\coloneqq (m\ell,m)\oplus kv_\ast \in U\oplus (-A_{m-1})$ where $\ell = \frac{(k^2-1)(m-1)}{2m}\in\Z$. By design, $[v_\ast'/m] = [kv_\ast/m] = [kb]$ and $v_\ast'^2 = v_\ast^2 = m(1-m)$.
    As we will show, given $v_\ast'$ one can always find a new basis for $B$ in which $v_\ast' = e_1'+2e_2'+\ldots+(m-1)e_{m-1}'$, where the $e_i'$ are a basis for $(-A_{m-1})$ with the same Gram matrix as the $e_i$.
    The map which sends $f_i$ to $f_i'$ and $e_i$ to $e_i'$ is then clearly an automorphism of $B$ which maps $v_\ast$ to $v_\ast'$, as desired.

    To this end, we proceed by extending $v_\ast'$ to a basis $\{v_\ast',\,e_2',\,e_3',\,\ldots,\,e_{m-1}',\,f_1',\,f_2'\}$ using the following iterative procedure.
    We can check that $v_\ast'\cdot x\equiv 0\mod{m}$ for all $x\in B$ and in fact one can easily find $x$ with $v_\ast'\cdot x = m$.
    Projecting this $x$ onto $v_\ast'^\perp\otimes\Q$ gives an element $\tilde{x} = x + \frac{v_\ast'}{m-1} \in (v_\ast'^\perp)^\vee$.
    Clearly $(m-1)\tilde{x}$ is the smallest multiple of $\tilde{x}$ in $v_\ast'^\perp$, so $G_{v_\ast'^\perp}\cong\Z/(m-1)\Z$.
    One quickly checks that $\tilde{x}^2 = \frac{m}{m-1}\mod{2}$ and so by the same logic as above $v_\ast'^\perp\cong U\oplus(-A_{m-2})$.
    Using this $U$, one can then find $y\in v_\ast'^\perp$ so as to make $e_2' \coloneqq x + y$ have norm $-2$.
    The Gram matrix for $\langle v_\ast',e_2'\rangle$ is then
    \begin{equation}
        \begin{bmatrix}
            m(1-m) & m \\ m & -2
        \end{bmatrix} \,,
    \end{equation}
    of determinant $m(m-2)$.
    Next, picking $e_3'$ is done in much the same way.
    $\langle v_\ast',e_2'\rangle$ cannot be an orthogonal summand, so we can find $x\in v_\ast'^\perp$ with $e_2'\cdot x = 1$ and the projection $\tilde{x} = x + \frac{v_\ast'}{m-2} + \frac{m-1}{m-2}e_2'$ of $x$ onto $\langle v_\ast',e_2'\rangle^\perp$ shows that $G_{\langle v_\ast',e_2'\rangle^\perp} \cong \Z/(m-2)\Z$ with generator of norm $\frac{m-1}{m-2}\mod{2}$ and therefore $\langle v_\ast',e_2'\rangle^\perp \cong U\oplus(-A_{m-3})$.
    Using this $U$, one can then find $y\in\langle v_\ast',e_2'\rangle^\perp$ such that $e_3'\coloneqq x+y$ has norm $-2$, giving the following Gram matrix for $\langle v_\ast',e_2',e_3'\rangle$:
    \begin{equation}
        \begin{bmatrix}
            m(1-m) & m \\ m & -2 & 1 \\ & 1 & -2
        \end{bmatrix} \,.
    \end{equation}
    This continues until all of $e_2',\ldots,e_{m-1}'$ have been chosen.
    At each intermediate step, the Gram matrix for $\langle v_\ast',e_2',\ldots,e_s'\rangle$ is
    \begin{equation}
        \begin{bmatrix}
            m(1-m) & m \\
            m & -2 & 1 \\
            & 1 & -2 \\
            & & & \ddots \\
            & & & & -2 & 1\\
            & & & & 1 & -2
        \end{bmatrix}_{s\times s}
    \end{equation}
    of determinant $(-1)^sm(m-s)$ and $\langle v_\ast',e_2',\ldots,e_s'\rangle^\perp \cong U\oplus(-A_{m-s-1})$.
    Once step $s=m-1$ is completed, $f_1',f_2'$ can be taken to be the standard basis for the remaining orthogonal $U$ summand.
    Finally, replacing $v_\ast'$ in the basis with
    \begin{equation}
        e_1' \coloneqq v_\ast' - 2e_2' - 3e_3' - \ldots - (m-1)e_{m-1}'
    \end{equation}
    returns the Gram matrix for the $(-A_{m-1})$ sublattice to the standard form~\eqref{eq:An-Gram} with $v_\ast' = e_1'+2e_2' +\cdots+ (m-1)e_{m-1}'$ as desired.
\end{proof}

The same strategy as above should work when $n<m$, but the form of $B$ clearly must be different.
We have found
\begin{eqnalign}
    m &\equiv 0\mod{8} \;: \quad &
        B &\cong \langle m\rangle \oplus (-E_8)^{(n-1)/8} \,,\\
    m &\equiv 1\mod{8} \;: &
        B &\cong \begin{bsmallmatrix}
            \frac{m-1}{4} & 1 \\ 1 & -4
        \end{bsmallmatrix} \oplus (-E_8)^{(n-2)/8} \,,\\
    m &\equiv 2\mod{8} \;: &
        B &\cong \begin{bsmallmatrix}
            \frac{m-2}{4} & 1 \\ 1 & -2
        \end{bsmallmatrix} \oplus \langle -2\rangle \oplus (-E_8)^{(n-3)/8} \,,\\
    m &\equiv 3\mod{8} \;: &
        B &\cong U\oplus\begin{bsmallmatrix}
            -\frac{m+1}{2} & 1 \\ 1 & -2
        \end{bsmallmatrix} \oplus (-E_8)^{(n-4)/8} \,,
\end{eqnalign}
again by simply checking that the signature and discriminant form are correct.
The proof of theorem~\ref{thm:orbits-even} can be co-opted essentially unchanged for $m\equiv 0\mod{8}$.
One could characterize $B$ for the remaining residue classes and analyze them all individually, but we have not attempted to do so.

The proof of theorem~\ref{thm:orbits-odd} breaks for $n=2$ because although lemma~\ref{lem:orbits-odd} still holds, the lattice $B$ has signature $(1,1)$ and is no longer necessarily uniquely fixed by its discriminant form since theorem~\ref{thm:Eichler} does not apply.
For the example of $v_1=(1;5,5)$ and $v_2=(1;1,7)$ mentioned at the beginning of this section, one finds 
\begin{equation}
    B_1 = v_1^\perp \cong \begin{bmatrix}
        0 & 7 \\ 7 & -2
    \end{bmatrix} \,, \quad
    B_2 = v_2^\perp \cong \begin{bmatrix}
        0 & 7 \\ 7 & -4
    \end{bmatrix}
\end{equation}
which are inequivalent over $\Z$ (e.g.\ $B_1$ represents $-2$ while $B_2$ does not), despite being equivalent over $\R$ and $\Z_p$ for all $p$.

\section{\texorpdfstring{$\eta$}{eta}-invariants}
\label{app:eta-invariants}

Here we provide some details for the computation of $\eta$-invariants for vector multiplets and hypermultiplets.
We consider an $\SU(2)^k$ subgroup of the group $G$ and turn on a background for $\U(1)^k\subset \SU(2)^k$.
Using the abbreviations $L^0 \coloneqq \calL_0$ and $L^q \coloneqq \calL_q\oplus\calL_{-q}$ for $q\geq 1$, a representation $R$ of $G$ that decomposes as
\begin{equation}
    R \to \bigoplus_{q\geq 0} n_q^R L^q
\end{equation}
has $\eta$-invariant
\begin{equation}
    \eta^R(L_p^7) = n_0^R\eta_0^\text{Dirac}(L_p^7) + \sum_{q\geq 1}n_q^R\big[\eta_q^\text{Dirac}(L_p^7) + \eta_{-q}^\text{Dirac}(L_p^7)\big] \,,
\end{equation}
where
\begin{equation}
    \eta_q^\text{Dirac}(L_p^7) = -\frac{(p^2 - 1)(p^2 + 11) + 30q(p - q)(q^2 - pq - 2)}{720p} + \Z
\end{equation}
as in equation~\eqref{eq:lens-space-eta}.
The decomposition of various representations of $G$ can be determined in a systematic way and we list the results for some common representations below.

\paragraph{\boldmath $\SU(N)$}

We use the following embedding:
\begin{equation}
    \SU(N)\supset \SU(N-2)\times \SU(2)\times \U(1)\supset \SU(N-2k)\times \SU(2)^k\times \U(1)^k\,,
\end{equation}
where $k\in\{1,2,\ldots, \lfloor N/2\rfloor\}$.
The branching rules for $\SU(N)\to \SU(N-2)\times \SU(2)\times \U(1)$ include
\begin{eqnalign}
    \rep{N} &\to (\rep{1},\rep{2})_{-N+1} \oplus (\rep{N-2},\rep{1})_2 \,,\\
    \rep{adj} &\to (\rep{1},\rep{1})_0 \oplus (\rep{1},\rep{3})_0 \oplus (\rep{N-2},\rep{2})_N \oplus (\repbar{N-2},\rep{2})_{-N} \oplus (\rep{adj},\rep{1})_0 \,,\\
    \rep{N(N-1)/2} &\to (\rep{1},\rep{1})_{-2N+4} \oplus (\rep{N-2},\rep{2})_{-N+4} \oplus (\rep{(N-2)(N-3)/2},\rep{1})_4 \,,\\
    \rep{N(N+1)/2} &\to (\rep{1},\rep{3})_{-2N+4} \oplus (\rep{N-2},\rep{2})_{-N+4} \oplus (\rep{(N-2)(N-1)/2},\rep{1})_4 \,.
\end{eqnalign}
From these and $\rep{2}\to (1)\oplus(-1)$, $\rep{3}\to(2)\oplus(0)\oplus(-2)$, etc.\ for $\SU(2)\to\U(1)$, we obtain
\begin{eqnalign}
    \rep{N} &\to k L \,,\\
    \rep{adj} &\to k^2 L^2 \oplus 2k(N-2k) L \,,\\
    \rep{N(N-1)/2} &\to \tfrac{k(k-1)}{2}L^2 \oplus k(N-2k)L \,,\\
    \rep{N(N+1)/2} &\to \tfrac{k(k+1)}{2} L^2 \oplus k(N-2k) L \,, \\
\end{eqnalign}
where here and below we leave $L^0$ unwritten to reduce clutter. Their number can be determined easily by comparing dimensions; for example the first line above stands for $\rep{N}\to kL \oplus (N-2k)L^0$.

\paragraph{\boldmath $\SO(2N)$}

We use the following embedding:
\begin{equation}
    SO(2N)\supset \SU(N)\times \U(1)\supset \SU(N-2k)\times \SU(2)^k\times \U(1)^{k+1}\,,
\end{equation}
The branching rules for $\SO(2N)\supset \SU(N)\times \U(1)$ include
\begin{eqnalign}
    \rep{2N} &\to \rep{N}_2 \oplus \repbar{N}_{-2} \,,\\
    \rep{adj} &\to \rep{1}_0 \oplus \rep{N(N-1)/2}_q \oplus \repbar{N(N-1)/2}_{-q} \oplus \rep{adj}_0 \,,\\
    \rep{(N+1)(2N-1)} &\to \rep{N(N+1)/2}_q \oplus \repbar{N(N+1)/2}_{-q} \oplus \rep{adj}_0 \,,
\end{eqnalign}
where $q=2$ for even $N$ and $q=4$ for odd $N$.
From these we obtain
\begin{eqnalign}
    \rep{2N} &\to 2kL \,,\\
    \rep{2N}\otimes\rep{2N} &\to 4k^2L^2 \oplus 8k(N-2k)L \,,\\
    \rep{adj} &\to k(2k+1)L^2 \oplus 4k(N-2k)L \,,\\
    \rep{(N+1)(2N-1)} &\to k(2k-1)L^2 \oplus 4k(N-2k)L \,.
\end{eqnalign}
The spinor representation is more involved.
As an example, take the spinor representation of $\SO(16)$.
Under $\SO(16)\to \SU(8)$, we obtain
\begin{eqnalign}
    \rep[+]{128} &\to 2\times\rep{8} \oplus 2\times\rep{56} \,,\\
    \rep[-]{128} &\to 2\times\rep{1} \oplus 2\times\rep{28} \oplus \rep{70} \,,
\end{eqnalign}
and thus
\begin{eqnalign}
    \rep[+]{128} &\to \tfrac{k(k-1)(k-2)}{3}L^3 \oplus k(k-1)(8-2k)L^2 \oplus k(5k^2-31k + 58)L \,,\\
    \rep[-]{128} &\to \tfrac{k(k-1)(k-2)(k-3)}{24}L^4 \oplus \tfrac{k(k-1)(k-2)}{3}(4-k)L^3 \\
    &\qquad\qquad \oplus k(k-1)\tfrac{7k^2 - 47k + 90}{6} L^2 \oplus k(4-k)\tfrac{7k^2 - 29k + 54}{3} L \,.
\end{eqnalign}
Similarly, for the spinor representation of $\SO(12)$, under $\SO(12)\to \SU(6)$ we have
\begin{eqnalign}
    \rep[+]{32} &\to 2\times\rep{6} \oplus \rep{20} \,,\\
    \rep[-]{32} &\to 2\times\rep{1} \oplus 2\times\rep{15} \,,
\end{eqnalign}
and
\begin{eqnalign}
    \rep[+]{32} &\to \tfrac{k(k-1)(k-2)}{6}L^3 \oplus k(k-1)(3-k)L^2 \oplus k\tfrac{5k^2-23k+34}{2}L \,,\\
    \rep[-]{32} &\to k(k-1)L^2 \oplus 4k(3-k)L \,.
\end{eqnalign}

\paragraph{\boldmath $\SO(2N+1)$}

This case is almost the same as for $\SO(2N)$:
\begin{equation}
    \SO(2N+1)\supset \SO(2N)\supset \SU(N)\times \U(1)\supset \SU(N-2k)\times \SU(2)^k\times \U(1)^{k+1}\,,
\end{equation}
where $k\in\{1,2,\ldots, \lfloor N/2\rfloor\}$.
We obtain
\begin{eqnalign}
    \rep{2N+1} &\to 2kL \,,\\
    \rep{adj} &\to k(2k-1)L^2 \oplus 2k(2N-4k+1)L \,,\\
    \rep{N(2N+3)} &\to k(2k+1)L^2 \oplus 2k(2N-4k+1)L \,.
\end{eqnalign}

\paragraph{\boldmath $\Sp(N)$}

We use the following embedding:
\begin{equation}
    \Sp(N)\supset \Sp(N-k)\times \SU(2)^k\,,
\end{equation}
where $k\in\{1,\ldots,N\}$.
The branching rules for $\Sp(N)\to \Sp(N-1) \times \SU(2)$ include
\begin{eqnalign}
    \rep{2N} &\to (\rep{1},\rep{2}) \oplus (\rep{2N-2},\rep{1}) \,,\\
    \rep{adj} &\to (\rep{1},\rep{1})\oplus(\rep{2N-2},\rep{2})\oplus(\rep{adj},\rep{1}) \,,\\
    \rep{N(2N+1)} &\to (\rep{1},\rep{3})\oplus (\rep{2N-2},\rep{2})\oplus(\rep{(N-1)(2N-1)},\rep{1}) \,.
\end{eqnalign}
We obtain
\begin{eqnalign}
    \rep{2N} &\to kL \,,\\
    \rep{adj} &\to \tfrac{k(k-1)}{2}L^2 \oplus (N-k)L \,,\\
    \rep{N(2N+1)} &\to \tfrac{k(k+1)}{2}L^2 \oplus (N-k)L \,.
\end{eqnalign}


\section{Column group data}
\label{app:column-groups}

Here we give a complete list of the column groups which appear in the linear program of section~\ref{sec:linear-programming}.
The $2961$ distinct column groups fall into $74$ families which differ only in the values of $m_i\geq 4$ for some of the rows.
In the following table listing these families, $m_i\in\{4,5,6,7,8,12\}$ unless otherwise restricted.

\begin{longtable}{*{7}{>$c<$}}
    \toprule
    \text{Column group} & \text{Restrictions} & \#(3) & \#(32) & \#(322) & \#(3222) & \text{Count} \\
    
    \midrule

    \begin{bsmallmatrix}
        1
    \end{bsmallmatrix}
    \begin{smallmatrix}
        m_1
    \end{smallmatrix} & & 0 & 0 & 0 & 0 & 6 \\[3pt]
    \begin{bsmallmatrix}
        1 \\ 1
    \end{bsmallmatrix}
    \begin{smallmatrix}
        m_1 \\ m_2
    \end{smallmatrix} & m_1 \geq m_2 & 0 & 0 & 0 & 0 & 21 \\[6pt]
    \begin{bsmallmatrix}
        1 \\ 1 \\ 1
    \end{bsmallmatrix}
    \begin{smallmatrix}
        m_1 \\ m_2 \\ m_3
    \end{smallmatrix} & m_1\geq m_2 \geq m_3 & 0 & 0 & 0 & 0 & 56 \\[9pt]
    \begin{bsmallmatrix}
        -1
    \end{bsmallmatrix}
    \begin{smallmatrix}
        m_1
    \end{smallmatrix} & & 0 & 0 & 0 & 0 & 6 \\[3pt]
    \begin{bsmallmatrix}
        -1 \\ 1
    \end{bsmallmatrix}
    \begin{smallmatrix}
        m_1 \\ m_2
    \end{smallmatrix} & & 0 & 0 & 0 & 0 & 36 \\[6pt]
    \begin{bsmallmatrix}
        -1 \\ 1 \\ 1
    \end{bsmallmatrix}
    \begin{smallmatrix}
        m_1 \\ m_2 \\ m_3
    \end{smallmatrix} & m_2 \geq m_3 & 0 & 0 & 0 & 0 & 126 \\[9pt]
    \begin{bsmallmatrix}
        -1 \\ 1 \\ 1 \\ 1
    \end{bsmallmatrix}
    \begin{smallmatrix}
        m_1 \\ m_2 \\ m_3 \\ m_4
    \end{smallmatrix} & m_2 \geq m_3 \geq m_4 & 0 & 0 & 0 & 0 & 336 \\[12pt]
    \begin{bsmallmatrix}
        2
    \end{bsmallmatrix}
    \begin{smallmatrix}
        m_1
    \end{smallmatrix} & & 0 & 0 & 0 & 0 & 6 \\[3pt]
    \begin{bsmallmatrix}
        -1 \\ 2
    \end{bsmallmatrix}
    \begin{smallmatrix}
        m_1 \\ m_2
    \end{smallmatrix} & m_2\geq 6 & 0 & 0 & 0 & 0 & 24 \\
    
    \midrule

    \begin{bsmallmatrix}
        -1 & 1 & 1
    \end{bsmallmatrix}
    \begin{smallmatrix}
        3
    \end{smallmatrix} & & 1 & 0 & 0 & 0 & 1 \\[3pt]
    \begin{bsmallmatrix}
        -1 & 1 & 1 \\
        2 & 1 & 1
    \end{bsmallmatrix}
    \begin{smallmatrix}
        3 \\ m_1
    \end{smallmatrix} & m_1 \geq 6 & 1 & 0 & 0 & 0 & 4 \\[6pt]
    \begin{bsmallmatrix}
        -1 & 1 & 1 \\
        1 & 1
    \end{bsmallmatrix}
    \begin{smallmatrix}
        3 \\ m_1
    \end{smallmatrix} & & 1 & 0 & 0 & 0 & 6 \\[6pt]
    \begin{bsmallmatrix}
        -1 & 1 & 1 \\
        & -1 & 1
    \end{bsmallmatrix}
    \begin{smallmatrix}
        3 \\ m_1
    \end{smallmatrix} & & 1 & 0 & 0 & 0 & 6 \\[6pt]
    \begin{bsmallmatrix}
        -1 & 1 & 1 \\
        2 & 1 & 1 \\
        & -1 & 1
    \end{bsmallmatrix}
    \begin{smallmatrix}
        3 \\ m_1 \\ m_2
    \end{smallmatrix} & m_1\geq 6 & 1 & 0 & 0 & 0 & 24 \\[9pt]
    \begin{bsmallmatrix}
        -1 & 1 & 1 \\
        1 & 1 \\
        1 & 1
    \end{bsmallmatrix}
    \begin{smallmatrix}
        3 \\ m_1 \\ m_2
    \end{smallmatrix} & m_1\geq m_2 & 1 & 0 & 0 & 0 & 21 \\[9pt]
    \begin{bsmallmatrix}
        -1 & 1 & 1 \\
        1 & 1 \\
        1 & & 1
    \end{bsmallmatrix}
    \begin{smallmatrix}
        3 \\ m_1 \\ m_2
    \end{smallmatrix} & m_1\geq m_2 & 1 & 0 & 0 & 0 & 21 \\[9pt]
    \begin{bsmallmatrix}
        -1 & 1 & 1 \\
        1 & 1 \\
        & -1 & 1
    \end{bsmallmatrix}
    \begin{smallmatrix}
        3 \\ m_1 \\ m_2
    \end{smallmatrix} & & 1 & 0 & 0 & 0 & 36 \\[9pt]
    \begin{bsmallmatrix}
        -1 & 1 & 1 \\
        1 & 1 \\
        & 1 & -1
    \end{bsmallmatrix}
    \begin{smallmatrix}
        3 \\ m_1 \\ m_2
    \end{smallmatrix} & & 1 & 0 & 0 & 0 & 36 \\[9pt]
    \begin{bsmallmatrix}
        -1 & 1 & 1 \\
        & -1 & 1 \\
        & 1 & -1
    \end{bsmallmatrix}
    \begin{smallmatrix}
        3 \\ m_1 \\ m_2
    \end{smallmatrix} & m_1\geq m_2 & 1 & 0 & 0 & 0 & 21 \\[9pt]
    \begin{bsmallmatrix}
        -1 & 1 & 1 \\
        2 & 1 & 1 \\
        & -1 & 1 \\
        & 1 & -1
    \end{bsmallmatrix}
    \begin{smallmatrix}
        3 \\ m_1 \\ m_2 \\ m_3
    \end{smallmatrix} & m_1\geq 6 \,,\; m_2\geq m_3 & 1 & 0 & 0 & 0 & 84 \\[12pt]
    \begin{bsmallmatrix}
        -1 & 1 & 1 \\
        1 & 1 \\
        1 & 1 \\
        & -1 & 1
    \end{bsmallmatrix}
    \begin{smallmatrix}
        3 \\ m_1 \\ m_2 \\ m_3
    \end{smallmatrix} & \begin{gathered}
        m_1 \geq m_2 \geq 5 \;\text{or}\\[-5pt]
        m_1\geq 7 \,,\; m_2 = 4
    \end{gathered} & 1 & 0 & 0 & 0 & 108 \\[12pt]
    \begin{bsmallmatrix}
        -1 & 1 & 1 \\
        1 & 1 \\
        1 & & 1 \\
        & -1 & 1
    \end{bsmallmatrix}
    \begin{smallmatrix}
        3 \\ m_1 \\ m_2 \\ m_3
    \end{smallmatrix} & & 1 & 0 & 0 & 0 & 216 \\[12pt]
    \begin{bsmallmatrix}
        -1 & 1 & 1 \\
        1 & 1 \\
        & -1 & 1\\
        & 1 & -1
    \end{bsmallmatrix}
    \begin{smallmatrix}
        3 \\ m_1 \\ m_2 \\ m_3
    \end{smallmatrix} & m_2,m_3\geq 5 & 1 & 0 & 0 & 0 & 150 \\[12pt]
    \begin{bsmallmatrix}
        -1 & 1 & 1 \\
        1 & 1 \\
        1 & & 1 \\
        & -1 & 1\\
        & 1 & -1
    \end{bsmallmatrix}
    \begin{smallmatrix}
        3 \\ m_1 \\ m_2 \\ m_3 \\ m_4
    \end{smallmatrix} & \scalebox{0.88}{$\begin{gathered}
        m_3>m_4\geq 5 \;\text{or}\\[-7pt]
        m_1\geq m_2 \,,\; m_3=m_4\geq 5
    \end{gathered}$} & 1 & 0 & 0 & 0 & 465 \\

    \midrule

    \begin{bsmallmatrix}
        -1 & 1 & & 1\\
        & -1 & 1 & 1
    \end{bsmallmatrix}
    \begin{smallmatrix}
        3 \\ 3'
    \end{smallmatrix} & & 2 & 0 & 0 & 0 & 1 \\[6pt]
    \begin{bsmallmatrix}
        -1 & 1 & & 1\\
        & -1 & 1 & 1\\
        2 & 1 & & 1
    \end{bsmallmatrix}
    \begin{smallmatrix}
        3 \\ 3' \\ m_1
    \end{smallmatrix} & m_1\geq 6 & 2 & 0 & 0 & 0 & 4 \\[9pt]
    \begin{bsmallmatrix}
        -1 & 1 & & 1\\
        & -1 & 1 & 1\\
        1 & 1 & 1
    \end{bsmallmatrix}
    \begin{smallmatrix}
        3 \\ 3' \\ m_1
    \end{smallmatrix} & & 2 & 0 & 0 & 0 & 6 \\[9pt]
    \begin{bsmallmatrix}
        -1 & 1 & & 1\\
        & -1 & 1 & 1\\
        1 & & -1 & 1
    \end{bsmallmatrix}
    \begin{smallmatrix}
        3 \\ 3' \\ m_1
    \end{smallmatrix} & & 2 & 0 & 0 & 0 & 6 \\[9pt]
    \begin{bsmallmatrix}
        -1 & 1 & & 1\\
        & -1 & 1 & 1\\
        1 & 1 & 1\\
        1 & & -1 & 1
    \end{bsmallmatrix}
    \begin{smallmatrix}
        3 \\ 3' \\ m_1 \\ m_2
    \end{smallmatrix} & & 2 & 0 & 0 & 0 & 36 \\

    \midrule

    \begin{bsmallmatrix}
        -1 & 1 & & 1 \\
        & -1 & 1 & 1 \\
        1 & & -1 & 1
    \end{bsmallmatrix}
    \begin{smallmatrix}
        3 \\ 3' \\ 3''
    \end{smallmatrix} & & 3 & 0 & 0 & 0 & 1\\[9pt]
    \begin{bsmallmatrix}
        -1 & 1 & & 1 \\
        & -1 & 1 & 1 \\
        1 & & -1 & 1 \\ 
        1 & 1 & 1
    \end{bsmallmatrix}
    \begin{smallmatrix}
        3 \\ 3' \\ 3'' \\ m_1
    \end{smallmatrix} & & 3 & 0 & 0 & 0 & 6\\

    \midrule

    \begin{bsmallmatrix}
        -1 & 1 & 1 \\
        & -1 & & 1
    \end{bsmallmatrix}
    \begin{smallmatrix}
        3 \\ 2
    \end{smallmatrix} & & 0 & 1 & 0 & 0 & 1\\[6pt]
    \begin{bsmallmatrix}
        -1 & 1 & 1 \\
        & -1 & & 1 \\
        2 & 1 & 1 & 1
    \end{bsmallmatrix}
    \begin{smallmatrix}
        3 \\ 2 \\ m_1
    \end{smallmatrix} & m_1 \geq 7 & 0 & 1 & 0 & 0 & 3\\[9pt]
    \begin{bsmallmatrix}
        -1 & 1 & 1 \\
        & -1 & & 1 \\
        1 & 1 & & 1
    \end{bsmallmatrix}
    \begin{smallmatrix}
        3 \\ 2 \\ m_1
    \end{smallmatrix} & & 0 & 1 & 0 & 0 & 6\\[9pt]
    \begin{bsmallmatrix}
        -1 & 1 & 1 \\
        & -1 & & 1 \\
        1 & & 1
    \end{bsmallmatrix}
    \begin{smallmatrix}
        3 \\ 2 \\ m_1
    \end{smallmatrix} & & 0 & 1 & 0 & 0 & 6\\[9pt]
    \begin{bsmallmatrix}
        -1 & 1 & 1 \\
        & -1 & & 1 \\
        & 1 & -1 & 1
    \end{bsmallmatrix}
    \begin{smallmatrix}
        3 \\ 2 \\ m_1
    \end{smallmatrix} & & 0 & 1 & 0 & 0 & 6\\[9pt]
    \begin{bsmallmatrix}
        -1 & 1 & 1 \\
        & -1 & & 1 \\
        1 & 1 & & 1 \\
        1 & & 1
    \end{bsmallmatrix}
    \begin{smallmatrix}
        3 \\ 2 \\ m_1 \\ m_2
    \end{smallmatrix} & & 0 & 1 & 0 & 0 & 36\\[12pt]
    \begin{bsmallmatrix}
        -1 & 1 & 1 \\
        & -1 & & 1 \\
        1 & & 1 \\
        1 & & 1
    \end{bsmallmatrix}
    \begin{smallmatrix}
        3 \\ 2 \\ m_1 \\ m_2
    \end{smallmatrix} & m_1\geq m_2 & 0 & 1 & 0 & 0 & 21\\[12pt]
    \begin{bsmallmatrix}
        -1 & 1 & 1 \\
        & -1 & & 1 \\
        1 & & 1 \\
        & 1 & -1 & 1
    \end{bsmallmatrix}
    \begin{smallmatrix}
        3 \\ 2 \\ m_1 \\ m_2
    \end{smallmatrix} & & 0 & 1 & 0 & 0 & 36\\[12pt]
    \begin{bsmallmatrix}
        -1 & 1 & 1 \\
        & -1 & & 1 \\
        1 & & 1 \\
        1 & & 1 \\
        & 1 & -1 & 1
    \end{bsmallmatrix}
    \begin{smallmatrix}
        3 \\ 2 \\ m_1 \\ m_2 \\ m_3
    \end{smallmatrix} & \begin{gathered}
        m_1,m_2,m_3\geq 5\,,\\[-7pt]
        m_1\geq m_2
    \end{gathered} & 0 & 1 & 0 & 0 & 75\\

    \midrule

    \begin{bsmallmatrix}
        -1 & & 1 & 1 \\
        1 & -1
    \end{bsmallmatrix}
    \begin{smallmatrix}
        3 \\ 2
    \end{smallmatrix} & & 0 & 1 & 0 & 0 & 1\\[6pt]
    \begin{bsmallmatrix}
        -1 & & 1 & 1 \\
        1 & -1 \\
        1 & 1 & 1
    \end{bsmallmatrix}
    \begin{smallmatrix}
        3 \\ 2 \\ m_1
    \end{smallmatrix} & & 0 & 1 & 0 & 0 & 6\\[9pt]
    \begin{bsmallmatrix}
        -1 & & 1 & 1 \\
        1 & -1 \\
        & & -1 & 1
    \end{bsmallmatrix}
    \begin{smallmatrix}
        3 \\ 2 \\ m_1
    \end{smallmatrix} & & 0 & 1 & 0 & 0 & 6\\[9pt]
    \begin{bsmallmatrix}
        -1 & & 1 & 1 \\
        1 & -1 \\
        1 & 1 & 1 \\
        1 & 1 & & 1
    \end{bsmallmatrix}
    \begin{smallmatrix}
        3 \\ 2 \\ m_1 \\ m_2
    \end{smallmatrix} & m_1\geq m_2 & 0 & 1 & 0 & 0 & 21\\[12pt]
    \begin{bsmallmatrix}
        -1 & & 1 & 1 \\
        1 & -1 \\
        1 & 1 & 1 \\
        & & -1 & 1
    \end{bsmallmatrix}
    \begin{smallmatrix}
        3 \\ 2 \\ m_1 \\ m_2
    \end{smallmatrix} & m_1\geq 5 & 0 & 1 & 0 & 0 & 30\\[12pt]
    \begin{bsmallmatrix}
        -1 & & 1 & 1 \\
        1 & -1 \\
        1 & 1 & 1 \\
        & & 1 & -1
    \end{bsmallmatrix}
    \begin{smallmatrix}
        3 \\ 2 \\ m_1 \\ m_2
    \end{smallmatrix} & & 0 & 1 & 0 & 0 & 36\\[12pt]
    \begin{bsmallmatrix}
        -1 & & 1 & 1 \\
        1 & -1 \\
        & & -1 & 1 \\
        & & 1 & -1
    \end{bsmallmatrix}
    \begin{smallmatrix}
        3 \\ 2 \\ m_1 \\ m_2
    \end{smallmatrix} & m_1\geq m_2 & 0 & 1 & 0 & 0 & 21\\[12pt]
    \begin{bsmallmatrix}
        -1 & & 1 & 1 \\
        1 & -1 \\
        1 & 1 & 1 \\
        1 & 1 & & 1 \\
        & & -1 & 1
    \end{bsmallmatrix}
    \begin{smallmatrix}
        3 \\ 2 \\ m_1 \\ m_2 \\ m_3
    \end{smallmatrix} & m_1,m_2\geq 5 & 0 & 1 & 0 & 0 & 150\\[15pt]
    \begin{bsmallmatrix}
        -1 & & 1 & 1 \\
        1 & -1 \\
        1 & 1 & 1 \\
        & & -1 & 1 \\
        & & 1 & -1
    \end{bsmallmatrix}
    \begin{smallmatrix}
        3 \\ 2 \\ m_1 \\ m_2 \\ m_3
    \end{smallmatrix} & m_1,m_2,m_3\geq 5 & 0 & 1 & 0 & 0 & 125\\[15pt]
    \begin{bsmallmatrix}
        -1 & & 1 & 1 \\
        1 & -1 \\
        1 & 1 & 1 \\
        1 & 1 & & 1 \\
        & & -1 & 1 \\
        & & 1 & -1
    \end{bsmallmatrix}
    \begin{smallmatrix}
        3 \\ 2 \\ m_1 \\ m_2 \\ m_3 \\ m_4
    \end{smallmatrix} & \scalebox{0.9}{$\begin{gathered}
        m_1,m_2,m_3,m_4\geq 5\,,\\[-7pt]
        m_1>m_2 \;\text{or} \\[-7pt]
        m_1=m_2\,,\;m_3\geq m_4
    \end{gathered}$} & 0 & 1 & 0 & 0 & 325\\

    \midrule

    \begin{bsmallmatrix}
        & 1 & -1 & 1 \\
        -1 & 1 & 1 \\
        & -1 & & 1
    \end{bsmallmatrix}
    \begin{smallmatrix}
        3 \\ 3' \\ 2'
    \end{smallmatrix} & & 1 & 1 & 0 & 0 & 1\\[9pt]
    \begin{bsmallmatrix}
        1 & & 1 & & -1 \\
        -1 & 1 & 1 \\
        & -1 & & 1
    \end{bsmallmatrix}
    \begin{smallmatrix}
        3 \\ 3' \\ 2'
    \end{smallmatrix} & & 1 & 1 & 0 & 0 & 1\\[9pt]
    \begin{bsmallmatrix}
        1 & & 1 & & -1 \\
        -1 & 1 & 1 \\
        & -1 & & 1 \\
        1 & & 1 & & 2
    \end{bsmallmatrix}
    \begin{smallmatrix}
        3 \\ 3' \\ 2' \\ m_1
    \end{smallmatrix} & m_1\geq 6 & 1 & 1 & 0 & 0 & 4\\[12pt]
    \begin{bsmallmatrix}
        1 & & 1 & & -1 \\
        -1 & 1 & 1 \\
        & -1 & & 1 \\
        1 & 1 & & 1 & 1
    \end{bsmallmatrix}
    \begin{smallmatrix}
        3 \\ 3' \\ 2' \\ m_1
    \end{smallmatrix} & m_1\geq 5 & 1 & 1 & 0 & 0 & 5\\

    \midrule

    \begin{bsmallmatrix}
        & & -1 & 1 & 1 \\
        -1 & & 1 & 1 \\
        1 & -1
    \end{bsmallmatrix}
    \begin{smallmatrix}
        3 \\ 3' \\ 2'
    \end{smallmatrix} & & 1 & 1 & 0 & 0 & 1\\[9pt]
    \begin{bsmallmatrix}
        & & -1 & 1 & 1 \\
        -1 & & 1 & 1 \\
        1 & -1 \\
        1 & 1 & 1 & & 1
    \end{bsmallmatrix}
    \begin{smallmatrix}
        3 \\ 3' \\ 2' \\ m_1
    \end{smallmatrix} & m_1\geq 5 & 1 & 1 & 0 & 0 & 5\\[12pt]
    \begin{bsmallmatrix}
        & & -1 & 1 & 1 \\
        -1 & & 1 & 1 \\
        1 & -1 \\
        1 & 1 & & 1 & -1
    \end{bsmallmatrix}
    \begin{smallmatrix}
        3 \\ 3' \\ 2' \\ m_1
    \end{smallmatrix} & & 1 & 1 & 0 & 0 & 6\\[12pt]
    \begin{bsmallmatrix}
        & & -1 & 1 & 1 \\
        -1 & & 1 & 1 \\
        1 & -1 \\
        1 & 1 & 1 & & 1 \\
        1 & 1 & & 1 & -1
    \end{bsmallmatrix}
    \begin{smallmatrix}
        3 \\ 3' \\ 2' \\ m_1 \\ m_2
    \end{smallmatrix} & m_1,m_2\geq 5 & 1 & 1 & 0 & 0 & 25\\

    \midrule

    \begin{bsmallmatrix}
        -1 & & 1 & 1 \\
        1 & -1 \\
        & & -1 & 1 & 1 \\
        & & & & -1 & 1
    \end{bsmallmatrix}
    \begin{smallmatrix}
        3 \\ 2 \\ 3' \\ 2'
    \end{smallmatrix} & & 0 & 2 & 0 & 0 & 1\\[12pt]
    \begin{bsmallmatrix}
        -1 & & 1 & 1 \\
        1 & -1 \\
        & & -1 & 1 & 1 \\
        & & & & -1 & 1 \\
        1 & 1 & 1 & & 1 & 1
    \end{bsmallmatrix}
    \begin{smallmatrix}
        3 \\ 2 \\ 3' \\ 2' \\ m_1
    \end{smallmatrix} & m_1\geq 6 & 0 & 2 & 0 & 0 & 4\\

    \midrule

    \begin{bsmallmatrix}
        -1 & 1 & 1 \\
        & -1 & & 1\\
        & & -1 & & 1
    \end{bsmallmatrix}
    \begin{smallmatrix}
        3 \\ 2 \\ 2
    \end{smallmatrix} & & 0 & 0 & 1 & 0 & 1\\[9pt]
    \begin{bsmallmatrix}
        -1 & 1 & 1 \\
        & -1 & & 1\\
        & & -1 & & 1 \\
        2 & 1 & 1 & 1 & 1
    \end{bsmallmatrix}
    \begin{smallmatrix}
        3 \\ 2 \\ 2 \\ m_1
    \end{smallmatrix} & m_1\geq 8 & 0 & 0 & 1 & 0 & 2\\[12pt]
    \begin{bsmallmatrix}
        -1 & 1 & 1 \\
        & -1 & & 1\\
        & & -1 & & 1 \\
        1 & 1 & & 1
    \end{bsmallmatrix}
    \begin{smallmatrix}
        3 \\ 2 \\ 2 \\ m_1
    \end{smallmatrix} & & 0 & 0 & 1 & 0 & 6\\[12pt]
    \begin{bsmallmatrix}
        -1 & 1 & 1 \\
        & -1 & & 1\\
        & & -1 & & 1 \\
        1 & 1 & & 1 \\
        1 & & 1 & & 1
    \end{bsmallmatrix}
    \begin{smallmatrix}
        3 \\ 2 \\ 2 \\ m_1 \\ m_2
    \end{smallmatrix} & \begin{gathered}
        m_1\geq m_2 \\[-7pt]
        (m_1,m_2)\neq (4,4)
    \end{gathered} & 0 & 0 & 1 & 0 & 20\\

    \midrule

    \begin{bsmallmatrix}
        -1 & & 1 & 1 \\
        1 & -1\\
        & & -1 & & 1
    \end{bsmallmatrix}
    \begin{smallmatrix}
        3 \\ 2 \\ 2
    \end{smallmatrix} & & 0 & 0 & 1 & 0 & 1\\[9pt]
    \begin{bsmallmatrix}
        -1 & & 1 & 1 \\
        1 & -1\\
        & & -1 & & 1 \\
        1 & 1 & 1 & & 1
    \end{bsmallmatrix}
    \begin{smallmatrix}
        3 \\ 2 \\ 2 \\ m_1
    \end{smallmatrix} & m_1\geq 5 & 0 & 0 & 1 & 0 & 5\\[12pt]
    \begin{bsmallmatrix}
        -1 & & 1 & 1 \\
        1 & -1\\
        & & -1 & & 1 \\
        1 & 1 & & 1
    \end{bsmallmatrix}
    \begin{smallmatrix}
        3 \\ 2 \\ 2 \\ m_1
    \end{smallmatrix} & & 0 & 0 & 1 & 0 & 6\\[12pt]
    \begin{bsmallmatrix}
        -1 & & 1 & 1 \\
        1 & -1\\
        & & -1 & & 1 \\
        & & -1 & & 1
    \end{bsmallmatrix}
    \begin{smallmatrix}
        3 \\ 2 \\ 2 \\ m_1
    \end{smallmatrix} & & 0 & 0 & 1 & 0 & 6\\[12pt]
    \begin{bsmallmatrix}
        -1 & & 1 & 1 \\
        1 & -1\\
        & & -1 & & 1 \\
        1 & 1 & 1 & & 1 \\
        1 & 1 & & 1
    \end{bsmallmatrix}
    \begin{smallmatrix}
        3 \\ 2 \\ 2 \\ m_1 \\ m_2
    \end{smallmatrix} & \scalebox{0.9}{$\begin{gathered}
        m_1\geq 5\,,\; m_2\geq 7 \;\text{or}\\[-7pt]
        m_1\geq 6\,,\; m_2\geq 5 \;\text{or}\\[-7pt]
        m_1\geq 8
    \end{gathered}$} & 0 & 0 & 1 & 0 & 25\\[15pt]
    \begin{bsmallmatrix}
        -1 & & 1 & 1 \\
        1 & -1\\
        & & -1 & & 1 \\
        1 & 1 & & 1 \\
        & & 1 & -1 & 1
    \end{bsmallmatrix}
    \begin{smallmatrix}
        3 \\ 2 \\ 2 \\ m_1 \\ m_2
    \end{smallmatrix} & m_1\geq 5 & 0 & 0 & 1 & 0 & 30\\

    \midrule

    \begin{bsmallmatrix}
        & & 1 & -1 & 1 \\
        -1 & & 1 & 1 \\
        1 & -1\\
        & & -1 & & 1
    \end{bsmallmatrix}
    \begin{smallmatrix}
        3 \\ 3' \\ 2' \\ 2'
    \end{smallmatrix} & & 1 & 0 & 1 & 0 & 1\\

    \midrule

    \begin{bsmallmatrix}
        -1 & & 1 & 1\\
        1 & -1\\
        & & -1 & & 1\\
        & & & -1 & & 1
    \end{bsmallmatrix}
    \begin{smallmatrix}
        3 \\ 2 \\ 2 \\ 2
    \end{smallmatrix} & & 0 & 0 & 0 & 1 & 1 \\[12pt]
    \begin{bsmallmatrix}
        -1 & & 1 & 1\\
        1 & -1\\
        & & -1 & & 1\\
        & & & -1 & & 1 \\
        1 & 1 & 1 & & 1
    \end{bsmallmatrix}
    \begin{smallmatrix}
        3 \\ 2 \\ 2 \\ 2 \\ m_1
    \end{smallmatrix} & m_1\geq 5 & 0 & 0 & 0 & 1 & 5 \\[15pt]
    \begin{bsmallmatrix}
        -1 & & 1 & 1\\
        1 & -1\\
        & & -1 & & 1\\
        & & & -1 & & 1 \\
        1 & 1 & 1 & & 1 \\ 
        1 & 1 & & 1 & & 1
    \end{bsmallmatrix}
    \begin{smallmatrix}
        3 \\ 2 \\ 2 \\ 2 \\ m_1 \\ m_2
    \end{smallmatrix} & \begin{gathered}
        m_1\geq m_2 \geq 6 \;\text{or}\\[-7pt]
        m_1\geq 8\,,\; m_2=5
    \end{gathered} & 0 & 0 & 0 & 1 & 12 \\
    \bottomrule
\end{longtable}

\bibliographystyle{JHEP}
\bibliography{refs}

\end{document}